\documentclass[a4paper,USenglish,cleveref,autoref,thm-restate]{lipics-v2021}

\usepackage[utf8]{inputenc}
\usepackage{graphicx}
\usepackage{enumerate}
\usepackage{url} 
\usepackage{hyperref}
\usepackage{amsthm,amssymb,amsfonts,amsmath}
\usepackage[ruled,linesnumbered]{algorithm2e}
\usepackage{xpatch}
\usepackage{bbm}
\usepackage{mathtools}

\usepackage{doi}
\usepackage{hyperref}
\usepackage{xcolor}
\usepackage{cite}
\usepackage{mathtools}

\graphicspath{{pictures/}{../pictures/}}

\newcommand{\GG}{\ensuremath{\mathcal{G}}\xspace}
\newcommand{\HH}{\ensuremath{\mathcal{H}}\xspace}
\newcommand{\KK}{\ensuremath{\mathcal{K}}\xspace}
\newcommand{\eps}{\varepsilon}

\DeclareMathOperator{\diam}{diam}
\DeclareMathOperator{\ecc}{ecc}
\DeclareMathOperator{\mean}{mean}
\DeclareMathOperator{\sumdist}{sumdist}
\DeclareMathOperator{\extreme}{endp}

\DeclareMathOperator{\CDF}{CDF}

\DeclareMathOperator{\Type}{Type}

\DeclareMathOperator{\area}{area}

\newcommand{\Create}{\ensuremath{\textsc{Create}}}
\newcommand{\Cut}{\ensuremath{\textsc{Cut}}}
\newcommand{\Link}{\ensuremath{\textsc{Link}}}
\newcommand{\GetVertexValue}{\ensuremath{\textsc{GetVertexValue}}}
\newcommand{\GetEdgeValue}{\ensuremath{\textsc{GetEdgeValue}}}
\newcommand{\AddTree}{\ensuremath{\textsc{AddTree}}}
\newcommand{\AddLeftPath}{\ensuremath{\textsc{AddLeftPath}}}
\newcommand{\MaxTree}{\ensuremath{\textsc{MaxTree}}}
\newcommand{\SumTree}{\ensuremath{\textsc{SumTree}}}
\DeclareMathOperator*{\Left}{\mathsf{left}}
\DeclareMathOperator*{\Right}{\mathsf{right}}

\def\sstart{\ensuremath{\mathsf{start}}}
\def\eend{\ensuremath{\mathsf{end}}}
\def\maxWeight{\ensuremath{\mathsf{maxWeight}}}
\def\extra{\ensuremath{\mathsf{extra}}}
\def\maxWeightPath{\ensuremath{\mathsf{maxWeightPath}}}
\def\extraPath{\ensuremath{\mathsf{extraPath}}} 
\def\maxWeightHanging{\ensuremath{\mathsf{maxWeightHanging}}} 
\def\extraHanging{\ensuremath{\mathsf{extraHanging}}}
\def\maxWeightLeft{\ensuremath{\mathsf{maxWeightLeft}}} 
\def\extraLeft{\ensuremath{\mathsf{extraLeft}}}
\def\maxWeightRight{\ensuremath{\mathsf{maxWeightRight}}} \def\extraRight{\ensuremath{\mathsf{extraRight}}}

\definecolor{KITblue}{rgb}{0.274 0.392 0.666}
\definecolor{KITblue50}{rgb}{0.637 0.696 0.833}
\definecolor{KITlilac}{rgb}{0.627 0 0.47}
\definecolor{KITlilac50}{rgb}{0.813 0.5 0.735}
\definecolor{KITgreen}{rgb}{0.509 0.745 0.235}

\definecolor{defblue}{rgb}{0.1,0.4,0.6} 
\let\emph\relax\DeclareTextFontCommand{\emph}{\color{defblue}\em}

\title{\texorpdfstring{Algorithms for Distance Problems\\ in Continuous Graphs}{Algorithms for Distance Problems in Continuous Graphs}}
\titlerunning{Algorithms for distance problems in continuous graphs}

\author{Sergio Cabello}{Faculty of Mathematics and Physics, University of Ljubljana, Ljubljana, Slovenia \and Institute~of~Mathematics, Physics and Mechanics, Ljubljana, Slovenia}{sergio.cabello@fmf.uni-lj.si}{0000-0002-3183-4126}{}

\author{Delia Garijo}{University of Seville, Spain}{dgarijo@us.es}{https://orcid.org/0000-0002-0493-4754}{}

\author{Antonia Kalb}{Technical University of Dortmund, Germany}{antonia.kalb@tu-dortmund.de}{https://orcid.org/0009-0009-0895-8153}{}

\author{Fabian Klute}{Universitat Politècnica de Catalunya, Barcelona, Spain}{fabian.klute@upc.edu}{0000-0002-7791-3604}{}

\author{Irene Parada}{Universitat Politècnica de Catalunya, Barcelona, Spain}{irene.parada@upc.edu}{0000-0003-3147-0083}{}

\author{Rodrigo I. Silveira}{Universitat Politècnica de Catalunya, Barcelona, Spain}{rodrigo.silveira@upc.edu}{https://orcid.org/0000-0003-0202-4543}{}

\authorrunning{S.\ Cabello, D.\ Garijo, A.\ Kalb, F.\ Klute, I.\ Parada, and R.\ I.\ Silveira} 
\Copyright{Sergio Cabello, Delia Garijo, Antonia Kalb, Fabian Klute, Irene Parada, and Rodrigo I.\ Silveira}


\relatedversion{A preliminary version of the paper appeared in 19th International Symposium on Algorithms and Data Structures (WADS 2025) \doi{10.4230/LIPIcs.WADS.2025.13}.} 

\funding{Delia Garijo, Fabian Klute, Irene Parada and Rodrigo I. Silveira are partially supported by grant PID2023-150725NB-I00 funded by MICIU/AEI/ 10.13039/501100011033. Delia Garijo is also partially supported by grants RED2024-153572-T and SOL2024-31596 (funded by PPIT-FEDER Andalucía 2021-2027). Sergio Cabello is funded in part by the Slovenian Research and Innovation Agency (P1-0297, N1-0218, N1-0285, J1-70045) and in part by the European Union (ERC, KARST, project number 101071836). Views and opinions expressed are however those of the authors only and do not necessarily reflect those of the European Union or the European Research Council. Neither the European Union nor the granting authority can be held responsible for them.}

\nolinenumbers 

\hideLIPIcs

\begin{document}

\maketitle

\begin{abstract}
    We study the problem of computing the diameter and the mean distance of a continuous graph, i.e., a  connected graph where all points along the edges, instead of only the vertices, must be taken into account. 
    It is known that for continuous graphs with $m$ edges these values can be computed in roughly $O(m^2)$ time. 
    In this paper, we use geometric techniques to obtain subquadratic time algorithms to compute the diameter and the mean distance of a continuous graph for two well-established classes of sparse graphs.
    We show that the diameter and the mean distance of a continuous graph whose treewidth is bounded by a constant $k$ can be computed in $O(n\log^{O(k)} n)$ time, where $n$ is the number of vertices in the graph.
    We also show that computing the diameter and the mean distance of a continuous planar graph with $n$ vertices and $F$ faces takes $O(n F \log n)$ time.
\end{abstract}

\section{Introduction}

Graph parameters dealing with distances provide fundamental information on the graph. 
The \emph{diameter}, defined as the maximum distance between any two vertices of a graph, and the \emph{mean distance}, which gives the average of all those distances, are natural concepts of great importance in real-world applications. While the  diameter gives the maximum eccentricity in the graph, the mean distance provides a measure of its compactness, and is closely related to the \emph{sum of the pairwise distances} of the graph and the well-known \emph{Wiener index}.\footnote{The \emph{sum of the pairwise distances} of a graph is the sum of distances between all ordered pairs of vertices and, for unweighted graphs, half of this value is the \emph{Wiener index}. This topological index has been studied extensively with thousands of publications.}

Computing the diameter and the sum of the pairwise distances of a given graph $G$ is a central problem in algorithmic graph theory.
A straightforward algorithm is to perform Dijkstra's algorithm from each vertex, allowing to compute both parameters in $O(nm+n^2\log n)$ time, where $n$ and $m$ are the number of vertices and edges of $G$, respectively. 
Given the high computational cost of this approach, considerable effort has been invested in developing faster algorithms, especially for sparse graphs (i.e., graphs with $m=O(n)$ edges).
It turns out that the general problem is notably difficult.
In 2013, Roditty and Vassilevska Williams showed that, for $\eps>0$ there is no $O(n^{2-\eps})$-time algorithm to compute the diameter of an arbitrary sparse graph unless the Strong Exponential Time Hypothesis (SETH) fails~\cite{RodittyW13}. Indeed, assuming the SETH, their proof shows that no $O(n^{2-\eps})$-time algorithm can distinguish between diameter $2$ or larger for unweighted sparse graphs. In addition,  one can deduce the same conditional lower bound for
computing the sum of the pairwise distances of the graph (see also~\cite{Cabello19}). 
This justifies the vast amount of ongoing research on identifying classes of sparse graphs for which these parameters can actually be computed in subquadratic time. 
Currently, such classes include graphs of bounded treewidth~\cite{Abboudetal2016,BHM20,CABELLO2009815}, graphs of bounded distance VC dimension~\cite{Ducoffe2022,Le2024}, median graphs~\cite{BergeDH24,BergeDH25}, and planar graphs~\cite{Cabello19,GawrychowskiKMS21}.

In this work, we tackle the challenge of subquadratic diameter and mean distance computation for \emph{continuous graphs} (these objects are also called  \textit{metric graphs} in other areas closer to analysis~\cite{BFb0086338,friedlander2005genericity}).
Our main motivation arises from \emph{geometric graphs}. A geometric graph is an undirected graph where each vertex is a two-dimensional point, and each edge is a straight line segment between the corresponding two points. 
These graphs have been widely studied not only for their theoretical interest but also for their numerous applications; they naturally arise in applications involving network design and geographic information, such as road or river networks; see, for instance,~\cite{PHILLIPS2015147,viana2013simplicity}.

The class of continuous graphs, formally defined in \cref{sec:preliminaries}, is actually more general than geometric graphs. Informally, the \emph{continuous graph} $\GG$ defined by an edge-weighted graph $G$ is the infinite set of points determined by the vertices and edges of $G$, where each point on an edge is considered part of the graph. Therefore, the graph $\GG$ can be considered an infinite set of points.  Every geometric graph can be seen as a continuous graph, but not the other way around. For example, a complete graph on four vertices can be seen as a continuous graph but it is not always realizable as a geometric graph (simply take all edge-lengths equal to $1$.)

The concept of distance extends from graphs  to continuous graphs in a natural way, as well as the notions of diameter and mean distance. However, there is a fundamental difference between the two settings: in the continuous case both parameters involve considering an infinite number of distances, those between any two points along the edges of the graph (rather than the distances only between the vertices, as in the discrete case.) 

\subparagraph{Related work.}

Distances in continuous graphs, especially the diameter,  have received a lot of interest recently, mainly in the context of  augmentation problems~\cite{ChenGar82,CaceresGHMPR18,DBLP:conf/swat/CarufelMS16,CGSS-17,GarijoMRS19,bae2019shortcuts,MeanDist23, GY23, GW22}.
See also~\cite{BKM24} for results on the mean distance in the context of geometric analysis.
Another well-known related problem about distances in graphs, also with a continuous aspect, is the computation of the absolute center of a graph, originally proposed by Hakimi~\cite{hakimi1964optimum}.

The diameter and the mean distance of a continuous graph $\GG$ with $n$ vertices and $m$ edges can be computed in $O(m^2+A(n,m))$ time~\cite{CaceresGHMPR18,ChenGar82,MeanDist23}\footnote{The algorithm in~\cite{CaceresGHMPR18} to compute the diameter is for plane geometric graphs, but it can be used also for arbitrary continuous graphs.}, where $A(n,m)$ is the  time\footnote{The currently best algorithm to compute all-pairs shortest paths for a graph with real weights has running time $A(n,m)=O(nm \log \alpha(m, n))$~\cite{pr-sparwug-05}, where $\alpha(m,n)$ is the extremely slowly growing inverse of the Ackermann function.
For some special graph classes faster algorithms are known, such as planar graphs with non-negative edge weights (where $A(m,n)=O(n^2)$~\cite{HENZINGER1997}), or graphs with integer non-negative edge weights (for which $A(n,m)=O(nm)$~\cite{T-99}).} required to compute all vertex-to-vertex distances in $\GG$.
For the diameter, this follows from the fact that, in a continuous graph, there always exists a \emph{diametral pair} (that is, a pair of points whose distance equals the diameter) that consists of either: (i) two vertices, (ii) two points on distinct non-pendant edges,\footnote{An edge $uv\in E(G)$ is {\em pendant} if either $u$ or $v$ is a {\em pendant} vertex (i.e., has degree 1).} or (iii) a pendant vertex and a point on a non-pendant edge~\cite[Lemma 6]{CaceresGHMPR18} (see \cref{fig:cases}). Regarding the mean distance, one can show that it is given by a weighted sum of the mean distances of all ordered pairs
of edges, which can be
obtained in constant time, once the distance matrix of the vertices of the graph has been
computed~\cite{MeanDist23}.

However, for sparse graphs, one hits again a quadratic running time barrier. 
Algorithms for diameter in discrete graphs do not carry over to continuous graphs, except in few situations (e.g., if there are only $O(1)$ different edge weights, then $O(1)$ Steiner points can be added to each edge, so that the diameter coincides with that of the continuous graph), and the same conditional lower bound of the discrete setting holds for the continuous case (one can reduce the continuous case to the discrete case by simply adding enough long paths to the graph.)

\begin{figure}
\centering
	\includegraphics{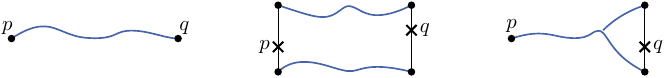}
     \caption{Types of diametral pairs of a continuous graph.}
	\label{fig:cases}
\end{figure}

The main challenge for continuous graphs is that the techniques that have successfully worked to speed-up the computation of the diameter and the sum of the pairwise distances for discrete graphs do not seem to easily extend.
The most similar setting to ours is perhaps that of planar graphs, for which recently the first subquadratic algorithms were discovered~\cite{Cabello19,GawrychowskiKMS21}.
These works use Voronoi diagrams in planar graphs to compute those values in the discrete setting. 
However, it is not clear whether they can be adapted to the continuous setting. More precisely, for a fixed source vertex and a fixed subgraph $H$ of a graph $G$, they compute the Voronoi diagram of $H$ using some additive weights. As the source moves, the additive weights defining the Voronoi diagram change, and the Voronoi diagrams change. Tracing those changes efficiently seems difficult, especially because the combinatorial structure of the Voronoi diagram may undergo important changes. Moreover, such changes can happen for several different movements of the sources. Thus, to achieve a subquadratic algorithm for planar continuous graphs, it seems that one should be able to treat those parallel changes in groups. The current technology for planar graphs does not seem ready for this.

\subparagraph{Contributions.}

In this work, we present subquadratic algorithms to compute the diameter and the mean distance for two
classes of sparse continuous graphs.
In fact, we consider the slightly more general framework of computing the diameter $\diam(\HH,\GG)$ and the mean distance $\mean(\HH,\GG)$ of a continuous subgraph $\HH\subseteq \GG$ \textit{with respect to the distance in} $\GG$;
precise definitions are given in \cref{subsec:parameters}.  
When $\HH=\GG$, we recover the usual diameter and mean distance of $\GG$ that we have been discussing until now.
This more general framework appears naturally in our algorithms for graphs of bounded treewidth during the recursion, but it seems interesting in its own right.
One may think of $\GG$ as the ambient space that defines the distances and of $\HH$ as the relevant subset of points of the space that we have to consider.

First, we study continuous graphs of bounded treewidth and
show how to compute their diameter and their mean distance in near-linear time.
The next two results distinguish whether the treewidth is assumed to be constant, as done in~\cite{CABELLO2009815},
or a parameter, as done in~\cite{BHM20}.  

\begin{theorem}[Combining \cref{thm:diameter-treewidth-1,thm:mean-treewidth-1}]
\label{thm:diameter-mean-treewidth-1}
	Let $k\ge 2$ be an integer constant, and let $\GG$ be the continuous graph defined by a graph $G$ with $n$ vertices, treewidth at most $k$, and
    nonnegative edge-lengths.
    Let $H$ be a subgraph of $G$ and let $\HH\subseteq \GG$
    be the corresponding continuous subgraph.
	The diameter $\diam(\HH,\GG)$ and the mean distance $\mean(\HH,\GG)$ can be computed in 
    $O(n \log^{4k-2} n)$ time.
\end{theorem}

\begin{theorem}[Combining \cref{thm:diameter-treewidth-2,thm:mean-treewidth-2}]
\label{thm:diameter-mean-treewidth-2}
	Let $\GG$ be the continuous graph defined by a graph $G$ with $n$ vertices, treewidth at most $k$, and
    nonnegative edge-lengths.
    Let $H$ be a subgraph of $G$ and let $\HH\subseteq \GG$
    be the corresponding continuous subgraph.
	The diameter 
	 $\diam(\HH,\GG)$ and the mean distance $\mean(\HH,\GG)$ can be computed in 
    $n^{1+\eps} 2^{O(k)}$ time, for any fixed $\eps>0$.
    \end{theorem}

Similarly to previous algorithms in the discrete setting to compute the diameter and the sum of the pairwise distances for graphs of bounded treewidth~\cite{Abboudetal2016,BHM20, CABELLO2009815}, the key technique that we use is orthogonal range searching; see also \cite{D22} for a novel application of this technique to compute the eccentricity and the distance-sum of any vertex of a directed weighted graph. However, compared to the discrete case, we need to use more dimensions in the orthogonal range searching because we do not need to handle only the interaction between vertices, but between pairs of edges.
For the mean distance, we also exploit that the mean distance between two edges can be interpreted as a sum of volumes of truncated prisms~\cite{MeanDist23}. In fact, we show that this mean distance can be expressed as a piecewise polynomial function of bounded degree whose variables encode the distance between the endpoints of the pair of edges. The evaluation of these functions can be done efficiently for all pairs of edges in bulks exploiting that data structures for orthogonal range searching give the answer to each query in a structured way.

Then, we turn our attention to planar graphs. 
By Euler's formula, the number $F$ of faces is the same for any embedding of a planar graph.
We obtain the following result about computing $\diam(\HH,\GG)$ and $\mean(\HH,\GG)$ 
of a continuous subgraph $\HH$ of $\GG$, when $\GG$ is planar.
Note that the time bound is subquadratic if $F = o(\frac{n}{\log n})$.

\begin{restatable}{theorem}{thmplan}
    \label{thm:planar}
    Let $\GG$ be the continuous graph defined by a planar graph $G$ with $n$ vertices, $F$ faces, and
    nonnegative edge-lengths.
	Let $H$ be a subgraph of $G$ and let $\HH\subseteq \GG$ be the corresponding continuous subgraph.
	The diameter $\diam(\HH,\GG)$ and the mean distance $\mean(\HH,\GG)$ can be computed in $O(nF\log n)$ time.
\end{restatable}

To obtain this result, we start fixing an embedding of an $n$-vertex planar graph $G$, obtaining a \emph{plane graph}.
For any face $f$ of $G$, let $\GG_f$ be the set of points on the boundary of $f$. 
We show how to represent in $O(n\log n)$ time the \emph{eccentricity} (i.e., the largest distance from a point) for all the points of $\GG_f$. 
Further, we show that the approach can be adapted to obtain, also in $O(n\log n)$ time, the mean distance from each point of $\GG_f$. By iterating over all faces, we obtain the claimed time bound of $O(nF\log n)$.

Previous techniques developed for the \textit{discrete} setting~\cite{CabelloCE13,DasKGW22,EricksonFL18,Klein05} can be easily adapted to compute in $O(n\log n)$ time, for each \textit{vertex} $v$ incident to a face $f$, the vertex of $G$ that is furthest from $v$. 
These algorithms dynamically maintain a shortest-path tree from the vertex $v$, as $v$ moves along $f$, and the vertices dynamically store the distance from the source as a label. Adding to this data structure a layer to report the vertex that is furthest from $v$, that is, the vertex with maximum label, the result for the discrete setting is obtained. A similar approach can be used for the mean distance
from a vertex incident to $f$; we just have to maintain the sum of the labels associated to all vertices.

To extend this idea to plane continuous graphs, we need to handle points on edges outside the shortest-path tree, because they may define the diameter,
and the source itself can be anywhere on the edges bounding $f$. For this extension, we cannot just use the approach of the discrete scenario in a black-box manner, but we have to change some of the low-level details.
As in some previous approaches, we start noting that the edges outside the shortest-path tree define a spanning tree of the dual graph
and the idea is to maintain distance information from a point $s\in \GG_f$, as we slide it along $\GG_f$.
To this end, we rely on two different dynamic graph data structures that support edge insertions and deletions.
The first data structure is for vertex-weighted forests, while the second one is for edge-weighted forests. 
The first one can be achieved with several existing data structures without modifications.
In contrast, for the second one we require an unusual operation not present in previous data structures for similar problems:
we need to maintain an embedded tree, instead of an abstract tree, and we have to be able to increase the weights for all the edges to the left of a given path.
We accomplish this by adapting \textit{top trees} to support this operation efficiently.

\subparagraph{Organization of the paper.}
The rest of the paper is organized as follows.
In \cref{sec:preliminaries} we define precisely continuous graphs, their diameter and their mean distance. 
We also explain some basic tools regarding treewidth and orthogonal range searching.  
In \cref{sec:diamtreewidth} we explain the computation of the diameter in graphs of bounded treewidth, while \cref{sec:meantreewidth} is devoted to the mean distance.
In \cref{sec:planar} we consider the computation of the diameter and the mean distance in planar graphs.
We conclude in \cref{sec:conclusion} with some remarks and open questions.


\section{Preliminaries}
\label{sec:preliminaries}

For a graph $G$ and a subset $X$ of its vertices, we use \emph{$G[X]$} to denote the subgraph of $G$ induced by $X$.
For each  positive integer $k$, we use the notation \emph{$[k]$} for $\{1,\dots,k\}$.

\subsection{Continuous graphs}
Let $G=(V(G),E(G),\ell)$ be an edge-weighted, connected graph,
where $\ell$ is a function that assigns a length $\ell(e)\ge 0$ to each edge $e \in E(G)$ (clearly, for our purposes we  assume that not all edge-lengths are 0). 
As mentioned before, the \emph{continuous graph} $\GG$ defined by $G$ is defined by the points along the edges of $G$. Since we do not regard $G$ as a graph drawn in the plane or in any other Euclidean space, we must define precisely what we mean by a point on an edge.

The idea is that a continuous graph $\GG$ can be thought of as a $1$-dimensional simplicial complex where each edge $uv$ is isometric to a segment of length $\ell(uv)$. Indeed, 
 each edge $uv$ of $G$ is associated with a closed segment  of length $\ell(uv)$, with the usual
metric and measure. This segment is denoted by $\GG(uv)$, and its endpoints are arbitrarily selected as  $\extreme(uv,u)$ and $\extreme(uv,v)$.
For each vertex $u$ of $G$, we glue (mathematically, we identify) 
all points $\extreme(uv,u)$ over all edges $uv$  incident to $u$;
the identified point is denoted by $\GG(u)$.
The continuous graph $\GG$ defined by $G=(V(G),E(G),\ell)$ is the resulting space, and
its \emph{total length $\ell(\GG)$} is given by $\sum_{uv\in E(G)} \ell(uv)$. 
In a continuous graph, each point has a small enough neighborhood 
that looks like a 1-dimensional interval or a star.

A point $p$ of $\GG$ can be specified by a triple 
$(uv,u,\lambda)\in E(G)\times V(G)\times [0,\ell(uv)]$, which represents the point of 
the segment $\GG(uv)$ at distance $\lambda$ (along $\GG(uv)$) from the endpoint $\extreme(uv,u)$; we use the notation $p\in \GG$.
The triples $(uv,u,\lambda)$ and $(uv,v,\ell(uv)-\lambda)$ define the same point of $\GG$.
Similarly, for any two incident edges $uv,uv'$ of $G$, the triples $(uv,u,0)$, $(uv,v,\ell(uv))$, 
$(uv',u,0)$ and $(uv',v',\ell(uv'))$ define the same point, namely $\GG(u)$.

We have already set the notation $\GG(uv)$ for an edge $uv$ and $\GG(u)$ for a vertex $u$.
With a slight abuse of notation, we do not distinguish $uv$ from $\GG(uv)$ and $u$ from $\GG(u)$. We extend this to subsets of edges $E'\subseteq E(G)$ where we do not necessarily distinguish the set $E'$ from $\GG(E')=\cup_{uv\in E'} \GG(uv)$.

Any subgraph $H$ of $G$ defines a continuous graph $\HH$, which can be seen as a subgraph of the continuous graph $\GG$ defined by $G$.

A \emph{walk} between two points $p,q$ of $ \GG$, or \emph{$pq$-walk}, is a sequence $pv_1, v_1v_2, \ldots, v_{k-1}v_k,v_kq$ where $v_1v_2, \ldots, v_{k-1}v_k\in E(G)$, the point $p$ lies on an edge incident to $v_1$, and the point $q$ lies on an edge incident to $v_k$. If the points $p,q$ are in the interior of the same edge, a $p$-to-$q$ walk might be given by a sequence without vertices. Notice that the difference with the discrete case is that walks in $\GG$ can contain fragments of edges (when $p$ or $q$ are points in the interior of the edges.)
If the first and the last point of the walk coincide, it is \emph{closed}. 
The \emph{length} of a walk is the sum of the length of its pieces, counted with multiplicity. 
For any two points $p,q\in \GG$, the \emph{distance} $d_{\GG}(p,q)$ is the minimum length 
over all $pq$-walks.
A \emph{shortest $pq$-path}  is 
a $pq$-walk $\pi(p,q)$ in $\GG$ such that $\ell(\pi(p,q))= d_{\GG}(p,q)$.
Such a walk is simple, meaning that it does not visit the same point of $\GG$ twice.

We note that the identification of edges of a continuous graph $\GG$ as segments of a certain length tells us that the distance between the two endpoints of a segment is not necessarily given by the length of the segment. However, the infinite set of points $\GG$ together with the distance function $d$ give rise to a metric space $(\GG,d)$.

It is easy to see that, for a point $p\in \GG$ represented by the triple $(uu',u,\lambda)$ 
and a point $p'\in \GG$ represented by $(vv',v,\lambda')$, it holds that
\begin{align*}
    d_{\GG}(p,p') ~=~ \min \{ &d_G(u,v)+\lambda+\lambda', ~ d_G(u,v')+\lambda+\ell(vv')- \lambda', \\
                               &d_G(u',v)+ \ell(uu')-\lambda+ \lambda', ~ d_G(u',v')+ \ell(uu')-\lambda+ \ell(vv')-\lambda'
                             \}.   
\end{align*}
We can also regard $d_\GG(p,q)$ as the discrete graph-theoretic distance in a graph obtained by subdividing $uu'$ with $p$ as a new vertex, subdividing $vv'$ with $q$ as a new vertex, and setting $\ell(up)=\lambda$, $\ell(pu')=\ell(uu')-\lambda$, $\ell(vq)=\lambda'$, and $\ell(qv')=\ell(vv')-\lambda'$.

\subsection{Distance-related descriptors}
\label{subsec:parameters}

Let $\GG$ be a continuous graph defined by a graph $G$, and let  $\HH$ be a continuous subgraph of $\GG$
defined by a subgraph $H\subseteq G$.
We are interested in distances between points of $\HH$ using the metric given by $\GG$.
One may think of $\GG$ as the ambient space and of $\HH$ as the relevant subset of the ambient space that we want to analyze.

The \emph{eccentricity} of a point $p\in \GG$ with respect to $\HH$ is 
$\ecc(p,\HH,\GG)=\max_{q \in \HH} d_{\GG}(p,q)$.
The \emph{diameter} of $\HH$ with respect to $\GG$ is defined as
$$\diam(\HH,\GG) =\max_{p,q\in \HH} d_{\GG}(p,q).$$ When $\HH=\GG$, we just 
talk about the eccentricity of $p$ in $\GG$ and the diameter of $\GG$, and write $\ecc(p,\GG)$ and $\diam(\GG)$ for $\ecc(p,\HH,\GG)$ and $\diam(\GG,\GG)$, respectively.
Clearly, $\diam(\HH,\GG)=\max_{p\in \HH}\ecc(p,\HH,\GG)$.

The \emph{sum of distances} of $\HH$ in $\GG$ is the sum of the pairwise distances in $\HH$, using the metric from $\GG$, that is, 
$$\sumdist(\HH,\GG)= \iint_{p,q\in \HH} d_{\GG}(p,q) \,dp \,dq.$$
(We use the term \textit{sum} because of its analogy to the discrete case, but it is of course an integral.)
The \emph{mean distance} of $\HH$ in $\GG$ is
$$\mean(\HH,\GG)= \frac{1}{\ell(\HH)^2}{\sumdist(\HH,\GG)}.$$
It is easy to see that 
$\mean(\HH,\GG)=\frac{1}{\ell(\HH)} \int_{p\in \HH} \mean(p,\HH,\GG) \,dp$,
where $$\mean(p,\HH,\GG)= \frac{1}{\ell(\HH)} \int_{q\in \HH} d_{\GG}(p,q) \,dq$$
is the mean distance from the point $p\in \HH$ with respect to $\HH$ in $\GG$.
When $\HH=\GG$, we just
write $\mean(\GG)$  and 
$\mean(p,\GG)$ for  $\mean(\GG,\GG)$ and $\mean(p,\GG,\GG)$, respectively.

\subsection{Treewidth}

The \emph{treewidth} of a graph is an important parameter in algorithmic graph theory which, roughly speaking, measures how far the graph is from being a tree. 
To define it, we first need to define tree decompositions.

A \emph{tree decomposition} of a graph $G$ is a pair $(X,T)$ where $X=\{X_i\subseteq V(G) \, | \, i\in I\}$ is a collection of subsets of $V(G)$, and $T$ is a tree with vertex set $I$ such that:
\begin{enumerate}
\item[(i)] $V(G)=\cup_{i\in I} X_i$;
\item[(ii)] for every edge $uv\in E(G)$ there is some $X_i$ such that $u,v\in X_i$;
\item[(iii)] for every $u\in V(G)$, the set of vertices $\{i\in I \, | \, u\in X_i\}$ induces a connected subtree of $T$.
\end{enumerate}

\noindent The \emph{width} of a tree decomposition is $\max_{i\in I} |X_i|-1$, and the \emph{treewidth} of $G$ is the minimum width among all tree decompositions of $G$.

Graphs with $n$ vertices and treewidth $k$ have $O(kn)$ edges~\cite{Bodlaender98}, and so graphs of bounded treewidth have $O(n)$ edges. Moreover, any graph $G$ with this property has a \emph{separation}, which is a triple $(A,B,S)$ such that $A,B,S\subset V(G)$, $A\cup B = V(G)$, $S=A\cap B$, and there is no edge incident to both $A\setminus B$ and $B\setminus A$.
The elements of $S$ in separation $(A,B,S)$ are called \emph{portals}.
Note that each path in $G$ from $A$ to $B$ must pass through a portal. See \cref{fig:separation}.

\begin{figure}[ht]
    \centering
    \includegraphics[page=10]{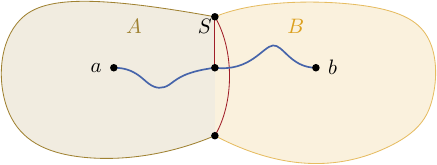}
    \caption{Visualization of a separation $(A,B,S)$;  portals of $S$ can be adjacent. Paths connecting any  vertex $a\in A$ with any vertex $b\in B$ (in blue) must pass through a portal.}
    \label{fig:separation}
\end{figure}

Informally, we are interested in separations where $A$ and $B$ have a constant fraction of the vertices and $S$ is small. 
Such separations of at most $k$ portals exist in graphs of treewidth $k$ and can be computed efficiently~\cite{BHM20, CABELLO2009815}.

\begin{lemma}[Cabello and Knauer~\cite{CABELLO2009815}]\label{theo:portals}
  Let $k > 1$ be a constant. Given a graph G with $n > k + 1$ vertices and 
treewidth at most $k$, we can find in linear time a separation $(A,B,S)$ in $G$ such that:
  \begin{enumerate}
    \item[\textup(i\textup)]  $A$ has between $\frac{n}{k+1}$ and $\frac{nk}{k+1}$ vertices;
     \item[\textup(ii\textup)]  $S$ has at most $k$ portals;
     \item[\textup(iii\textup)]  adding edges between the portals in $S$ does not change the treewidth of $G$.
  \end{enumerate}
\end{lemma}

For simplicity, when using \cref{theo:portals}, we will assume that $S$ contains exactly $k$ portals; it may happen that it has fewer. In addition, 
throughout this work, we will consider two different settings, depending on whether the treewidth is assumed to be constant, as done in~\cite{CABELLO2009815}, or a parameter, as done in~\cite{BHM20}. 

\subsection{Orthogonal range searching}\label{sec:ors}

A \emph{rectangle} in $\mathbb{R}^d$ is the Cartesian product of $d$  intervals (whose extremes can be included or not). Orthogonal range searching allows us to preprocess a point set $P$ from $\mathbb{R}^d$ in such a way that, for a query rectangle $R$, certain properties of $P\cap R$ can be efficiently reported. The properties may also take into account labels attached to the points.

We use the notation $B(t,d)=\binom{d+\lceil \log t\rceil}{d}$ to bound the performance of the data structures we will use, where $t$ refers to the number of points in $P$. First, we recall the following asymptotic bounds.

\begin{lemma}[Bringmann, Husfeldt, and Magnusson~\cite{BHM20}]
    \label{lem:Bnd}
    $B(t,d)=O(\log^d t)$ and $B(t,d)=t^{\eps}2^{O(d)}$ for each $\eps>0$.
\end{lemma}

The analysis of orthogonal range searching performed in~\cite[Section 3]{BHM20} (see also~\cite{M80}) 
leads to the data structure described in the following theorem. 
We use the version suggested by Cabello~\cite{Cabello22} because it adapts better to our needs.
(The notation $\bigsqcup$ indicates that the union is between pairwise disjoint sets.)

\begin{theorem}[Cabello~\cite{Cabello22}]
\label{thm:rangesearching}
	Given a set $P$ of $t$ points in $\mathbb{R}^d$, there is a family of 
	sets $\mathcal{P}=\{ P_j \mid j\in J\}$ and a data structure with
	the following properties: 
	\begin{itemize}
		\item $P_j\subseteq P$ for each $P_j\in \mathcal{P}$;
		\item all the sets of $\mathcal {P}$ together have $O(td\cdot B(t,d))$ points, 
			i.e., $\sum_{P_j\in \mathcal{P}} |P_j|=O(td\cdot B(t,d))$;
		\item for each query rectangle $R\subset \mathbb{R}^d$,
			the data structure finds in $O(2^d B(t,d))$ time indices $J_R\subset J$ such that 
			$|J_R|= O(2^d B(t,d))$
			and $P\cap R = \bigsqcup_{j\in J_R} P_j$;
		\item the family $\mathcal{P}$ and the data structure 
			can be computed in $O(td\cdot B(t,d))$ time.
	\end{itemize}
\end{theorem}

\section{Diameter in continuous graphs of bounded treewidth}
\label{sec:diamtreewidth}

In this section, we discuss the computation of the diameter of a continuous graph $\GG$ defined by a graph $G$ with $n$ vertices, $m$ edges and treewidth $k$.

Note that computing the diameter of continuous trees (i.e., treewidth $1$) can be reduced to the discrete setting, 
because in trees the diameter is always attained by two vertices. 
Thus, we restrict our attention to $k>1$.
As in~\cite{Abboudetal2016,BHM20,Cabello19,CABELLO2009815}, we use orthogonal range searching to work with distance-related problems in graphs of bounded treewidth. However, there is a fundamental difference: since our graphs are continuous graphs, we have to deal with pairs of edges instead of pairs of vertices, and the interaction between edges is much more complex. 

Sometimes we need to consider an orientation for each edge of $G$, to distinguish its vertices. The edges of $G$ are oriented arbitrarily, but we 
keep track of the orientation. When the orientation is not relevant, an edge is denoted by $uv$; otherwise, it is written as $(u,v)$ or $(v,u)$, depending on its orientation. In both cases, the edge set of $G$ is denoted by $E(G)$.

\subsection{Characterization of the diameter via walks}\label{subsec:walks}

We begin by characterizing the diameter in the continuous setting in terms of the lengths of certain walks (compare to \cref{fig:cases}, which shows types of diametral pairs according to their position on the edges.)

For each pair of edges $a_0a_1,b_0b_1 \in E(G)$, let \emph{$W(a_0a_1,b_0b_1)$}
be a shortest closed walk 
passing through all the interior points of $a_0a_1$ and $b_0b_1$.

\begin{lemma}\label{lem:closed_walks}
    Let $\GG$ be a continuous graph defined by a graph $G=(V(G),E(G),\ell)$. 
    For any two distinct edges $a_0a_1,b_0b_1 \in E(G)$, we have
    $$\max_{p\in a_0a_1, \, q\in b_0b_1} d_{\GG}(p,q)=\frac{\ell(W(a_0a_1,b_0b_1))}{2}.$$
\end{lemma}

\begin{proof} 
    Let $p'\in a_0a_1$ and $q'\in b_0b_1$ be two points such that $d_{\GG}(p',q')$ attains the maximum distance over all points $p\in a_0a_1$ and $ q\in b_0b_1$.
    Let $P=a_0u_1\ldots u_kb_0$ be a shortest $a_0b_0$-path. Refer to \cref{fig:closed_walks}(left).
    Exchanging the roles of $a_0$ and $a_1$, if needed, and the roles of $b_0$ and $b_1$, if needed, 
    we assume without loss of generality that $\ell(p'a_0)+\ell(P)+\ell(b_0q')=d_{\GG}(p',q')$. 
    This also covers the cases where $p'$ or $q'$ are vertices.
    Let $P'=a_1u'_1\ldots u'_{k'}b_1$ be a shortest $a_1b_1$-path. It holds that
    \[
       \ell(p'a_1)+\ell(P')+\ell(b_1q') ~=~ d_{\GG}(p',q') ~=~ \ell(p'a_0)+\ell(P)+\ell(b_0q').
    \]
    Indeed, if $\ell(p'a_1)+\ell(P')+\ell(b_1q')<d_{\GG}(p',q')$, then there would be a $p'q'$-path shorter than $d_{\GG}(p',q')$, which is not possible.
    If $\ell(p'a_1)+\ell(P')+\ell(b_1q')>d_{\GG}(p',q')$, then there would be two points on, respectively, $a_0a_1$ and $b_0b_1$ with larger distance than $d_{\GG}(p',q')$, again a contradiction with the fact that 
    $d_{\GG}(p',q')$ is the maximal distance over all points on  the edges.
    
    Now, the sequence $p'a_0u_1\ldots u_kb_0q'b_1u'_{k'}\ldots u'_1a_1p'$ defines a closed walk $W'$ with length $2 d_{\GG}(p',q')$ and passing through all the interior points of $a_0a_1$ and $b_0b_1$. Therefore, 
    \[
       d_{\GG}(p',q') ~=~ \frac{\ell(W')}{2} ~\geq~ \frac{\ell(W(a_0a_1,b_0b_1))}{2},
    \]
   since $W(a_0a_1,b_0b_1)$ is a shortest closed walk containing the two edges.
    On the other hand, any closed walk passing through $p'$ and $q'$ has length at least $2 d_{\GG}(p',q')$,
    and so $\ell(W(a_0a_1,b_0b_1))\geq 2 d_{\GG}(p',q')=\ell(W')$.
    Thus,   
    \[
       \max_{p\in a_0a_1, \, q\in b_0b_1} d_{\GG}(p,q)=d_{\GG}(p',q') ~=~ \frac{\ell(W')}{2} ~=~ \frac{\ell(W(a_0a_1,b_0b_1))}{2}.
    \qedhere
     \]
\end{proof}

\begin{figure}
    \begin{minipage}{.45\linewidth}
    \centering
       \includegraphics[page=1]{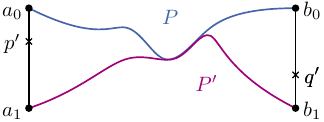}
    \end{minipage}\hfill
    \begin{minipage}{.45\linewidth}
    \centering
   \includegraphics[page=2]{treewidth-figs.pdf}
    \end{minipage}    
    \caption{Visualization of the setting in the proof of \cref{lem:closed_walks}.\\
    Left: By combining a shortest $a_0b_0$-path $P$, a shortest $a_1b_1$-path $P'$ and the edges  $a_0a_1$ and $b_0b_1$, one obtains a shortest closed walk passing through all the interior points of $a_0a_1$ and $b_0b_1$.
    Right: The distance $d_{\GG}(p',q')$ is given by the length of the path $p'a_1b_0q'$ (taking the diagonal edge), but the closed walk $W(a_0a_1,b_0b_1)$ is the outer cycle $a_0a_1b_1b_0$.}
    \label{fig:closed_walks}
\end{figure}

\cref{fig:closed_walks}(right) shows that taking the maximum in \cref{lem:closed_walks} is necessary: the distance between two points is not given, in general, by half the length of a shortest closed walk containing the two edges where the points are located.
Note that, in general, \cref{lem:closed_walks} is not true when $a_0a_1=b_0b_1$,
unless we assume that the shortest path between points of an edge is contained in the edge.

Noting that \cref{lem:closed_walks} is only for distinct edges, the 
following corollary is immediate.

\begin{corollary}\label{theo:diam-walks}
    If $H\subseteq G$ defines the continuous subgraph $\HH$ of $\GG$,
    then $\diam(\HH,\GG)$ is the maximum between
    \[\max_{\begin{minipage}{2.3cm}\scriptsize\centering
            $a_0a_1,b_0b_1\in E(H)$\\ 
            $a_0a_1\neq b_0b_1$\end{minipage}}
        \frac{\ell(W(a_0a_1,b_0b_1))}{2}
        ~~\text{ and }~~
        \max_{e\in E(H)} \max_{p,q\in e} d_{\GG}(p,q).
    \]
\end{corollary}

For each pair of distinct edges $a_0a_1, b_0b_1\in E(G)$,
\cref{lem:lengthW} below establishes that there are two candidate values for
$\ell(W(a_0a_1,b_0b_1))$.

\begin{lemma}\label{lem:lengthW}
    For each distinct $a_0a_1, b_0b_1\in E(G)$, it holds that
    \[
        \ell(W(a_0a_1,b_0b_1)) = \ell(a_0a_1)+\ell(b_0b_1) + 
            \min\{ d_G(a_0,b_0)+ d_G(a_1,b_1) , \, d_G(a_0,b_1)+ d_G(a_1,b_0)\}.
    \]
\end{lemma}

\begin{proof}
    Set 
    \begin{align*}       
        \ell_1(a_0a_1,b_0b_1) ~&=~ \ell(a_0a_1)+ \ell(b_0b_1) + d_G(a_0,b_0)+d_G(a_1,b_1)~~ \text{ and}\\
        \ell_2(a_0a_1,b_0b_1) ~&=~ \ell(a_0a_1)+\ell(b_0b_1)+ d_G(a_0,b_1) +d_G(a_1,b_0).
    \end{align*}
    Both quantities are the length of some closed walk passing through
    all the interior points of $a_0a_1$ and $b_0b_1$.
    Therefore, $$\ell(W(a_0a_1,b_0b_1)) \le \min \{ \ell_1(a_0a_1,b_0b_1),\ell_2(a_0a_1,b_0b_1)\}.$$

    Consider now the optimal walk $W(a_0a_1,b_0b_1)$ and orient it
    such that we go along $a_0a_1$ from $a_0$ to $a_1$. In this direction, 
    we may walk $b_0b_1$ from $b_0$ to $b_1$ or in the opposite direction. 
    In the first case, the length of $W(a_0a_1,b_0b_1)$ is at least $\ell_2(a_0a_1,b_0b_1)$, 
    while in the second case the length is at least $\ell_1(a_0a_1,b_0b_1)$.
    In either case, we have $\ell(W(a_0a_1,b_0b_1)) \ge \min\{ \ell_1(a_0a_1,b_0b_1),\ell_2(a_0a_1,b_0b_1)\}$.
\end{proof}

As in the proof of \cref{lem:lengthW}, for our computations, we might be interested in considering  distinct, oriented edges $(a_0,a_1),(b_0,b_1)\in E(G)$. 
We distinguish the two cases of \cref{lem:lengthW}  by saying that $(a_0,a_1,b_0,b_1)$ is of \emph{type $1$} if
\[
    \ell(W(a_0a_1,b_0b_1)) = \ell_1(a_0a_1,b_0b_1)= \ell(a_0a_1)+\ell(b_0b_1) + 
        d_G(a_0,b_0)+ d_G(a_1,b_1),
\]    
and of \emph{type $2$} otherwise (see \cref{fig:pathtypes}(left)).
For each oriented edge $(b_0,b_1)\in E(G)$ and type $\tau\in \{1,2\}$, we define 
\begin{align}\label{eq:type12}
    \Type_\tau(b_0,b_1) &= \{ (a_0,a_1)\in E(G)\mid (a_0,a_1,b_0,b_1) \text{ is of type $\tau$}\}.
\end{align}
Therefore, 
\[
    (a_0,a_1)\in \Type_1(b_0,b_1) ~\Longleftrightarrow~
    d_G(a_0,b_0)+ d_G(a_1,b_1) \le d_G(a_0,b_1)+ d_G(a_1,b_0).
\]

\begin{figure}[ht]
    \centering
    \includegraphics[page=4]{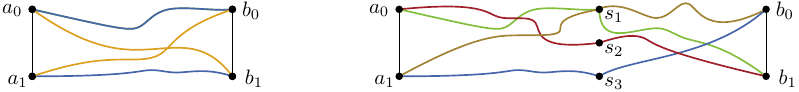}
    \caption{Left: The shape of closed walks $W(a_0a_1,b_0b_1)$ for $(a_0,a_1,b_0,b_1)$ of type 1 (in blue) and of type 2 (in orange). Right: Shortest paths determined by vertices $a_0, a_1, b_0,b_1$ and the portals involved; the paths between $a_0b_1$ (in green) and $a_1b_0$ (in brown) go through the same portal, $s_1$.}
    \label{fig:pathtypes}
\end{figure}

\subsection{Diameter across the portals}
\label{sec:diam-across}
Let $(A,B,S)$ be a separation in $G$ with $k$ portals, and fix an enumeration $s_1,\dots,s_k$ of them. 
Let $E_A\subseteq E(G[A])$ and $E_B\subseteq E(G[B])$, where $E(G[A])$ and $E(G[B])$ are the edge sets of the subgraphs of $G$ induced by, respectively, $A$ and $B$. 
We further assume that $E_A$ and $E_B$ are disjoint.

Any path in $G$ connecting a vertex  $a\in A$ with a vertex $b\in B$ must pass through a portal, but we are interested in the first portal 
in the enumeration that lies in a shortest $ab$-path (notice that $a$ or $b$ could already be some portal of $S$.) In order to codify this information, we next define the function $\varphi(i;a,b)$.

For each index $i\in [k]$, each vertex $a\in A$, and each vertex $b\in B$,
let $\varphi(i;a,b)$ be the logic predicate that holds whenever 
$s_i$ is the first portal in the enumeration that lies in some shortest path from $a$ to $b$.
Formally, 
\begin{align*}
    \varphi(i;a,b) ~=~ &\bigwedge_{j<i}
            \big[ d_G(a,b) < d_G(a,s_j) + d_G(s_j,b)\big] 
            ~\wedge~ \big[d_G(a,b) = d_G(a,s_i) + d_G(s_i,b)\big]. 
\end{align*}
It is easy to see that, for each $(a,b)\in A\times B$, 
there exists a unique index $i\in [k]$ where $\varphi(i;a,b)$ holds; in other words, $|\{ i\in [k] \mid \varphi(i;a,b) \}|=1$. As an example, in \cref{fig:separation}, there are three portals, which can be enumerated, say, by decreasing $y$-coordinate in the drawing; assuming that the blue path is a shortest $ab$-path and  there is no shortest path via $s_1$, we  have $\varphi(2;a,b)$.

We  extend this idea to the four shortest paths determined by any two 
vertices $a_0,a_1\in A$ and any two vertices $b_0,b_1\in B$. Refer to \cref{fig:pathtypes}(right). There are now at most four portals involved, and we are interested in the shortest way of connecting edges in $E_A$ with edges in $E_B$ via the portals, generating shortest closed walks of the form $W(a_0a_1,b_0b_1)$. 
This is codified in the following predicate that is defined for all
$\kappa=(i_{0,0},i_{0,1},i_{1,0},i_{1,1})\in [k]^4$, all $a_0,a_1\in A$, and all $b_0,b_1\in B$:
\begin{align*}\label{def:phi}
\Phi(\kappa; a_0,a_1,b_0,b_1) =
    \Phi((i_{0,0},i_{0,1},i_{1,0},i_{1,1}); a_0,a_1,b_0,b_1) =
    \bigwedge_{(\alpha,\beta)\in \{0,1\}^2}
    \varphi(i_{\alpha,\beta}; a_\alpha,b_\beta).
\end{align*}
Thus, $\Phi(\kappa; a_0,a_1,b_0,b_1)$ holds if and only if,
for each $\alpha,\beta \in \{ 0,1\}$, the index $i_{\alpha,\beta}$
is the smallest index $i$ with the property that $s_i$
lies in some shortest path from $a_\alpha$ to $b_\beta$.
As before, for each $(a_0,a_1,b_0,b_1)\in A^2\times B^2$, 
there exists a unique $4$-tuple
$\kappa\in [k]^4$ where $\Phi(\kappa; a_0,a_1,b_0,b_1)$ holds.

Now, \cref{lem:closed_walks} tells us that we must obtain the maximum length 
over all shortest closed walks of the form $W(a_0a_1,b_0b_1)$. 
Here it is relevant that the edge sets $E_A$ and $E_B$ are disjoint;
compare to \cref{theo:diam-walks}. Thus,
for each $\kappa\in [k]^4$, each oriented edge $(b_0,b_1)\in E_B$,
and each type $\tau\in \{1,2 \}$ 
we define
\begin{align*}
    \hspace{-3mm}\Lambda_\tau(\kappa;b_0,b_1)= \max \{ \ell(W(a_0a_1,b_0b_1)) \mid (a_0,a_1)\in E_A\cap \Type_\tau(b_0,b_1) \wedge \Phi(\kappa;a_0,a_1,b_0,b_1)  \}.
\end{align*}
This represents the maximum length over all edges $(a_0, a_1)$  of a type-$\tau$ closed shortest walk containing edges $(b_0, b_1)$ and $(a_0, a_1)$, and consistent with the portals in $\kappa$.

We next discuss how to compute efficiently $\Lambda_\tau(\kappa;b_0,b_1)$ 
for several edges $(b_0,b_1)\in E_B$ simultaneously. Recall the notation $B(t,d)=\binom{d+\lceil \log t\rceil}{d}$ introduced in \cref{sec:ors}.

\begin{lemma}\label{le:accross}
    Consider a fixed type $\tau\in \{1,2\}$ and indices $\kappa\in [k]^4$.
    In $O(m 2^{4k-3} B(m,4k-3))$ time we can compute
    the values $\Lambda_\tau(\kappa;b_0,b_1)$ for all $(b_0,b_1)\in E_B$.
\end{lemma}
\begin{proof}
    Consider a fixed $\kappa=(i_{0,0},i_{0,1},i_{1,0},i_{1,1})\in [k]^4$.
    When possible, we omit from the notation the dependency on $\kappa$. 

    For each vertex $a\in A$ and each $i\in [k]$, we define the point 
    $p(i;a) \in \mathbb{R}^k$ whose $j$-th coordinate is 
    $p_j(i;a) = d_G(a,s_i)-d_G(a,s_j)$.
    For each vertex $b\in B$ and each index $i\in [k]$, we define the rectangle 
    $R(i;b)=I_1(i;b)\times\dots\times I_k(i;b)\subset \mathbb{R}^k$,
    where $I_j(i;b)$ is the interval
    \[
    	I_j(i;b) ~=~ \begin{cases}
    			\bigl( -\infty,\, d_G(b,s_j)-d_G(b,s_i) \bigr) & \text{if $j<i$,}\\
    			\mathbb{R} & \text{if $j=i$,}\\
    			\bigl( -\infty,\, d_G(b,s_j)-d_G(b,s_i) \bigr] & \text{if $j>i$.}
    		\end{cases}
    \]
    In $p(i;a)$ and $I(i;b)$ we remove the $i$-th coordinate, which has value zero or does not provide any information, respectively. 
    We use the same notation
    for the resulting objects, now in $\mathbb{R}^{k-1}$.
    It has been shown~\cite{CABELLO2009815,BHM20,Cabello22} 
    that the point $p(i;a)$ lies in the rectangle $R(i;b)$ if and only if 
    $\varphi(i;a,b)$ holds. Now, we do something similar for $\Phi(\kappa; a_0,a_1,b_0,b_1)$.
    
    For any two vertices $a_0,a_1\in A$, we define the point $p(a_0,a_1)$
    by concatenating the coordinates of
    $p(i_{0,0};a_0)$, $p(i_{0,1};a_0)$, $p(i_{1,0};a_1)$, and $p(i_{1,1};a_1)$.
    Thus, the point $p(a_0,a_1)$ has $4k-4$ coordinates.
    In addition, for any two vertices $b_0,b_1\in B$ we make
    a rectangle $R(b_0,b_1)$ in $\mathbb{R}^{4k-4}$
    by taking the Cartesian product of
    $R(i_{0,0};b_0)$, $R(i_{0,1};b_1)$, $R(i_{1,0};b_0)$, and $R(i_{1,1};b_1)$. 
    It holds that
    \begin{align*}
        p(a_0,a_1)\in R(b_0,b_1) ~~\Longleftrightarrow~~
        & \forall \alpha,\beta \in \{ 0,1\}:~~ 
            p(i_{\alpha,\beta};a_\alpha)\in R(i_{\alpha,\beta};b_\beta)\\
        ~~\Longleftrightarrow~~
        & \forall \alpha,\beta \in \{ 0,1\}:~~ 
            \varphi(i_{\alpha,\beta};a_\alpha, b_\beta)  \\
        ~~\Longleftrightarrow~~
        & \Phi(\kappa; a_0,a_1,b_0,b_1).
    \end{align*}

We construct a shortest-path tree from each portal $s_i$ and 
    store the distance $d_G(v,s_i)$ to each vertex $v$ of $G$.
    This preprocessing takes $O(k(m+n\log n))$ time using standard
    shortest-path algorithms. This running time is bounded by 
    $O(m 2^{4k-3} B(m,4k-3))$ because $4k-3 \ge 5$ for $k\ge 2$.
    After this, any distance from any portal to any vertex can be
    accessed in constant time. This means that we can compute the $O(k)$ coordinates of $p(a_0,a_1)$ 
    for all $(a_0,a_1)\in E_A$ and the 
    $O(k)$ intervals of the rectangles $R(b_0,b_1)$ for all 
    $(b_0,b_1)\in E_B$ in $O(k|E_A|+ k|E_B|) = O(mk)$ time. 
    
Now, the edges $(a_0,a_1)$ considered in the definition of $\Lambda_\tau(\kappa;b_0,b_1)$, are determined, in particular, by their type $\tau$ with respect to $(b_0,b_1)$.  To capture this information, for $\tau=1$ (similar for $\tau=2$), we append to each point $p(a_0,a_1)$
    an extra coordinate given by
    \[
        p_{\rm extra}(a_0,a_1)=d_G(a_0,s_{i_{0,0}})+d_G(a_1,s_{i_{1,1}}) - d_G(a_0,s_{i_{0,1}}) - d_G(a_1,s_{i_{1,0}}).
    \]
    Let $p^+(a_0,a_1)$ be the resulting point in $4k-3$ dimensions.
    Similarly, we extend the rectangle $R(b_0,b_1)$ by  an
    extra dimension by taking the Cartesian product with the interval
    \[
        I_{\rm extra}(b_0,b_1) = \big(-\infty, ~d_G(s_{i_{0,1}}, b_1) + d_G(s_{i_{1,0}}, b_0) - d_G(s_{i_{0,0}}, b_0) - d_G(s_{i_{1,1}}, b_1) \big].
    \]
    Let $R^+(b_0,b_1)$ be the resulting rectangle in $4k-3$ dimensions.
    Computing these points and intervals takes additional $O(m)$ time.

    We have that $p^+(a_0,a_1) \in R^+(b_0,b_1)$ if and only if 
    $\Phi(\kappa;a_0,a_1,b_0,b_1)$  
    and $p_{\rm extra}(a_0,a_1) \in I_{\rm extra}(b_0,b_1)$ holds. 
    To explain the rationale behind the previous expressions, first note that the property $\Phi(\kappa;a_0,a_1,b_0,b_1)$
    translates into  
    \begin{quote}
        $\forall (\alpha,\beta)\in \{0,1\}^2$:~~ $i_{\alpha,\beta}$ is the smallest index $i$ such that $s_i$ lies in some shortest path
        from $a_\alpha$ to $b_\beta$. 
    \end{quote}
    In turn, the condition $p_{\rm extra}(a_0,a_1)\in I_{\rm extra}(b_0,b_1)$
    is equivalent to 
    \begin{align*}
        d_G(a_0,s_{i_{0,0}}) + &d_G(s_{i_{0,0}}, b_0) + 
     d_G(a_1,s_{i_{1,1}}) + d_G(s_{i_{1,1}}, b_1) \le \\
     &d_G(a_0,s_{i_{0,1}}) + d_G(s_{i_{0,1}}, b_1) +
     d_G(a_1,s_{i_{1,0}}) + d_G(s_{i_{1,0}}, b_0),
    \end{align*}
    which, under the assumption that 
    $\Phi(\kappa;a_0,a_1,b_0,b_1)$ holds, becomes
    \[ d_G(a_0,b_0)+ d_G(a_1,b_1) \le d_G(a_0,b_1)+ d_G(a_1,b_0). \]
    This is precisely the condition for $(a_0,a_1,b_0,b_1)$ to be of type 1.
    We conclude that
    \[
        p^+(a_0,a_1) \in R^+(b_0,b_1) ~~\Longleftrightarrow ~~
            (a_0,a_1)\in \Type_1(b_0,b_1) \text{ and }
            \Phi(\kappa;a_0,a_1,b_0,b_1).
    \]
    Therefore,
    \begin{equation}\label{eq:ranges}
        \forall (b_0,b_1)\in E_B:~~ \Lambda_1(\kappa;b_0,b_1) = 
        \max \{ \ell(W(a_0a_1,b_0b_1)) \mid p^+(a_0,a_1) \in R^+(b_0,b_1) \}.   
    \end{equation}
    
    We use orthogonal range searching to handle the right side 
    of \cref{eq:ranges}, as follows.
    
    Let
    $P_A=\{ p^+(a_0,a_1) \mid (a_0,a_1)\in E_A\}$. 
    These are $|E_A|=O(m)$ points in $\mathbb{R}^{4k-3}$.
    For $P_A$ we compute $\mathcal{P}=\{ P_j \mid j\in J\}$ 
    and the data structure of \cref{thm:rangesearching}. 
    This takes $O(|P_A| (4k-3) \cdot B(|E_A|,4k-3)) = 
    O(m k \cdot B(m,4k-3))$ time.
    For each $j\in J$, we compute and store the weight
    \[
        \omega_j=\max\{ d_G(a_0,s_{i_{0,0}}) + \ell(a_0a_1) + d_G(a_1,s_{i_{1,1}}) \mid 
        p^+(a_0,a_1) \in P_j\}.
    \]
    Computing these values $\omega_j$ over all $j\in J$ takes
    $O(\sum_j |P_j|)$ and does not affect the asymptotic running time.
    This finishes the description of the construction of the data structure.

    Consider now an edge $(b_0,b_1)\in E_B$.
    We query the data structure with the rectangle $R=R^+(b_0,b_1)$
    and get the indices $J_R$ such that 
    $P_A\cap R^+(b_0,b_1) = \bigsqcup_{j\in J_R} P_j$.
    The query takes $O(2^{4k-3} B(m,4k-3))$ time
    and $|J_R|=O(2^{4k-3} B(m,4k-3))$.
    Next, we compute in $O(2^{4k-3} B(m,4k-3))$ time
    \begin{align*}
        \omega(b_0,b_1) &= \max \{ \omega_j\mid j\in J_R \} \\
                        &= \max \{ d_G(s_{i_{0,0}},a_0) + \ell(a_0a_1) +
                        d_G(a_1,s_{i_{1,1}}) 
            \mid p^+(a_0,a_1)\in R^+(b_0,b_1)\}.
    \end{align*}
    Now, we can compute the value $\Lambda_1(\kappa;b_0,b_1)$
    using that
    \begin{align*}
        \Lambda_1(\kappa;b_0,b_1) &= \max \{ \ell(W(a_0a_1,b_0b_1)) \mid (a_0,a_1)\in \Type_1(b_0,b_1) \wedge \Phi(\kappa;a_0,a_1,b_0,b_1)  \}\\
        &= \max \{ \ell(a_0a_1)+\ell(b_0b_1) + d_G(a_0,b_0)+ d_G(a_1,b_1) \mid\\ &~~~~~~~~~~~~~~~~~~~~~~~~~~~~
            (a_0,a_1)\in \Type_1(b_0,b_1) \wedge \Phi(\kappa;a_0,a_1,b_0,b_1)  \}\\
        &= d_G(s_{i_{0,0}},b_0) + \ell(b_0b_1) + d_G(b_1,s_{i_{1,1}}) \\
        & ~~~~+ \max \{d_G(s_{i_{0,0}},a_0) + \ell(a_0a_1) + d_G(a_1,s_{i_{1,1}}) \mid
        \\
        &~~~~~~~~~~~~~~~~~~~~~~~~~~~~(a_0,a_1)\in \Type_1(b_0,b_1) \wedge \Phi(\kappa;a_0,a_1,b_0,b_1) \}\\
        &= d_G(s_{i_{0,0}},b_0) + \ell(b_0b_1) + d_G(b_1,s_{i_{1,1}}) \\
        & ~~~~+ \max \{d_G(s_{i_{0,0}},a_0) + \ell(a_0a_1) + d_G(a_1,s_{i_{1,1}}) \mid p^+(a_0,a_1)\in R^+(b_0,b_1)  \}\\
        &= d_G(s_{i_{0,0}},b_0) + \ell(b_0b_1) + d_G(b_1,s_{i_{1,1}}) +
            \omega(b_0,b_1).        
    \end{align*}
    This last step takes $O(1)$ time.
    In total, we spend $O(2^{4k-3} B(m,4k-3))$ time per edge 
    $(b_0,b_1)$ to compute $\Lambda_1(\kappa;b_0,b_1)$.

    The computation of $\Lambda_2(\kappa;b_0,b_1)$ is analogous to the one described for $\tau=1$, and therefore we focus on the differences between both types.   
    For $\tau=2$, we define the interval  
    \[
        \tilde I_{\rm extra}(b_0,b_1) = \big(d_G(s_{i_{0,1}}, b_1) + d_G(s_{i_{1,0}}, b_0) - d_G(s_{i_{0,0}}, b_0) - d_G(s_{i_{1,1}}, b_1), ~ +\infty \big),
    \]
    and define the rectangle $\tilde R^+(b_0,b_1) = R(b_0,b_1)\times \tilde I_{\rm extra}(b_0,b_1)$. 
    (Note that $I_{\rm extra}(b_0,b_1) \sqcup \tilde I_{\rm extra}(b_0,b_1)=\mathbb{R}$,
    which is expected because the two types are complementary.)
    
    We then have 
    \[
        p^+(a_0,a_1) \in \tilde R^+(b_0,b_1) ~~\Longleftrightarrow ~~
            (a_0,a_1)\in \Type_2(b_0,b_1) \text{ and }
            \Phi(\kappa;a_0,a_1,b_0,b_1).
    \]
    We keep using the same data structure for the set of points $P_A$.
    However, for each $j\in J$, we compute another weight,
    namely
    \[
        \tilde\omega_j=\max\{ d_G(a_0,s_{i_{0,1}}) + \ell(a_0a_1) + d_G(a_1,s_{i_{1,0}}) \mid 
        p^+(a_0,a_1) \in P_j\}.
    \]
    When considering the edge $(b_0,b_1)\in E_B$, we query the data
    structure with $\tilde R= \tilde R^+(b_0,b_1)$ and get a 
    set $J_{\tilde R}\subset J$
    of $O(2^{4k-3} B(m,4k-3))$ indices
    such that 
    $P_A\cap \tilde R^+(b_0,b_1)= \bigsqcup_{j\in J_{\tilde R}} P_j$.
    We then compute
    \begin{align*}
        \tilde\omega(b_0,b_1) &= \max \{ \tilde\omega_j\mid j\in J_{\tilde R} \}.
    \end{align*}
    and note that 
    $\Lambda_2(\kappa;b_0,b_1) = d_G(s_{i_{1,0}},b_0) + \ell(b_0b_1) + 
            d_G(b_1,s_{i_{0,1}}) + \tilde \omega(b_0,b_1)$.
\end{proof}

We are now ready to compute the maximum distance over all pairs of points $p,q$ lying, respectively, in the continuous graphs determined by the edge sets $E_A$ and $E_B$. Thus, the $pq$-paths must go through the portals (note that points $p,q$ could be located in the continuous graph $\GG_S$ defined by $G[S]$, since $S=A\cap B$.)   

\begin{lemma}
\label{lem:AB}
    If $E_A$ and $E_B$ are disjoint, then
    we can compute $\displaystyle\max_{p\in \GG(E_A), q\in \GG(E_B)} d_{\GG}(p,q)$ in 
    $O(m 2^{4k-3} k^4 B(m,4k-3))$ time.
\end{lemma}
\begin{proof}
    We use \cref{le:accross} to compute for each type $\tau\in \{1,2\}$ and indices $\kappa\in [k]^4$ the values $\Lambda_\tau(\kappa;b_0,b_1)$ for all $(b_0,b_1)\in E_B$.
    This requires applying \cref{le:accross} a total of $O(k^4)$ times, and therefore we spend 
    $O(m 2^{4k-3} k^4 B(m,4k-3))$ time.
    
    Because for each $(a_0,a_1)\in E_A$ and each $(b_0,b_1)\in E_B$,  there
    is precisely one $\kappa\in [k]^4$ such that $\Phi(\kappa;a_0,a_1,b_0,b_1)$,
    and because the types complement each other,
    we have 
    \[
        \forall (b_0,b_1)\in E_B:~~~ \max_{(a_0,a_1)\in E_A} \ell(W(a_0a_1,b_0b_1))  = \max \big\{ \Lambda_\tau(\kappa;b_0,b_1)\mid \kappa\in [k]^4, ~\tau\in \{1,2\} \big\}.
    \]
    Using \cref{lem:closed_walks,lem:lengthW}, we note that 
    \begin{align*}
    \max_{p\in \GG(E_A), q\in \GG(E_B) } d_{\GG}(p,q) 
    &= \max_{(a_0, a_1)\in E_A,~ (b_0, b_1)\in E_B}~~ \max_{p\in a_0a_1, q\in b_0b_1} d_{\GG}(p,q) \\
    &= \max_{(a_0, a_1)\in E_A,~ (b_0, b_1)\in E_B}~~\frac{\ell(W(a_0a_1,b_0b_1))}{2} \\
    &= \tfrac{1}{2} \cdot\max_{(b_0, b_1)\in E_B}~~ \max_{ (a_0,a_1)\in E_A}{\ell(Wa_0a_1,b_0b_1))} \\
    &= \tfrac{1}{2} \cdot \max_{(b_0, b_1)\in E_B}~~ \max \big\{ \Lambda_\tau(\kappa;b_0,b_1)\mid \kappa\in [k]^4, ~\tau\in \{1,2\} \big\}\\
     &= \tfrac{1}{2} \cdot \max \big\{ \Lambda_\tau(\kappa;b_0,b_1)\mid 
     (b_0,b_1)\in E_B,~\kappa\in [k]^4, ~\tau\in \{1,2\} \big\} .
    \end{align*}
    With this expression, it is clear we can obtain the result in
    $O(|E_B|\cdot k^4\cdot 2)$ additional time.
\end{proof}

\begin{remark*}
We are aware that some log factor can
be shaved off, and that for small $k$ one can improve the analysis slightly. 
However, we prefer to maintain this high-level structure to keep things simpler and parallel to the forthcoming computation of mean distance.
\end{remark*}

\subsection{Global diameter}
\label{sec:global-diameter}
Let $G$ be a graph with $n$ vertices and $m$ edges defining
the continuous graph $\GG$. Let $H$ be a subgraph of $G$
defining a continuous graph $\HH\subseteq \GG$.
We use a divide-and-conquer approach to compute $\diam(\HH,\GG)$. 

Let $(A,B,S)$ be a separation in $G$. 
For $X\in \{ A,B\}$, let $\GG_{X}$ be the continuous subgraph of $\GG$
defined by $G[X]$ and let $\HH_{X}=\HH\cap \GG_X$. 
Thus, $\HH_X$ is the continuous subgraph of $\HH$ defined by $H[X\cap V(H)]$.
We define $\HH_C$ as the continuous subgraph of $\HH_B$ obtained from $\HH_B$ after removing $\HH_A\setminus S$. Equivalently, $\HH_C$ is the continuous subgraph of $\HH$ defined by 
$H[B]-E(H[A])= H[B]-E(H[S])$. Thus, $\HH_C$ contains the edges of $\HH_B$ that have at least one vertex in
$B\setminus A=B\setminus S$, but it does not contain edges of $\HH_B$ between portals of $S$ (that is, with both endpoints in $S$); these edges are considered in $\HH_A$.
Observe that $\HH=\HH_A\cup \HH_C$ and $\HH_A\cap \HH_C\subseteq S$. We emphasize that $\HH_A$ and $\HH_C$ share only vertices of $S$ and no edges of $\HH$. This property is essential to apply Lemma~\ref{lem:AB} in this section, and to avoid double counting in the computation of the mean distance in Section~\ref{sec:global-mean}. 

We have
\begin{align*}
    \diam(\HH,\GG)&=\max \Big\{\max_{p,q\in \HH_A} d_{\GG}(p,q) ,~ \max_{p,q\in \HH_C} d_{\GG}(p,q), \displaystyle \max_{p\in \HH_A,\, q\in \HH_C} d_{\GG}(p,q) \Big\}.
\end{align*}

One must be careful because, for $X\in \{ A,B\}$, the shortest paths between two points $p,q\in \HH_X$ may not be contained in $\GG_X$: they may traverse the portals to the opposite side and then return. Thus, we might have $d_{\GG}(p,q)< d_{\GG_X}(p,q)$ for some $p,q\in \HH_X$. This is an issue to do divide-and-conquer, which requires treating each recursive subproblem independently. 
We handle this by adding extra edges between some portals, as follows.

Let $E'$ be the subset of pairs of portals $ss'\in \binom{S}{2}$ such that $d_G(s,s')< d_{G[A]}(s,s')$ or $d_G(s,s')< d_{G[B]}(s,s')$.
Let $G'$ be the graph obtained from $G$ by adding, for each pair $ss'\in E'$, an edge $ss'$ with length $d_G(s,s')$. If the edge $ss'$ already exists in $G$, we allow parallel edges in $G'$: the original edge $ss'\in E(G)$ whose length is $\ell(ss')$, and the new edge $ss'$ of length $d_G(s,s')$. Here, we slightly abuse the notation by not using $\ell(ss')$ to denote the length of this new ``artificial'' edge, and instead simply indicate its value. 
For $X\in \{A,B\}$, let $\GG'_X$  the continuous graph defined 
by $G'[X]$.
The edges $ss'\in E'$  guarantee that 
$d_{\GG}(p,q)=d_{\GG'_X}(p,q)$ for any two points $p, q\in \GG_{X}$ 
and $X\in\{A,B \}$.
Therefore, since $\HH_C\subseteq \GG_B$, we have (see \cref{fig:divide_and_conquer_diam})
\begin{align}
    \diam(\HH,\GG)&=\max \Big\{\max_{p,q\in \HH_A} d_{\GG}(p,q) ,~ \max_{p,q\in \HH_C} d_{\GG}(p,q), \displaystyle \max_{p\in \HH_A,\, q\in \HH_C} d_{\GG}(p,q) \Big\}\notag\\
    &=\max \Big\{\max_{p,q\in \HH_A} d_{\GG'_A}(p,q) ,~ \max_{p,q\in \HH_C} d_{\GG'_B}(p,q), \displaystyle \max_{p\in \HH_A, \, q\in \HH_C} d_{\GG}(p,q) \Big\}\notag\\
    &= \max \Big\{\diam(\HH_A,\GG'_A),~ \diam(\HH_C,\GG'_B), \displaystyle \max_{p\in \HH_A, \, q\in \HH_C} d_{\GG}(p,q)\Big\}.\label{for:dc}
\end{align}
The last term can be computed by \cref{lem:AB} because the edge sets defining $\HH_A$
and $\HH_C$ are disjoint, but the running time depends on the size of $S$. 
The first two terms can be computed recursively.
Note that through the recursion, each edge may be added artificially to keep the distances correct at most once; after adding $ss'\in E'$ with length $d_G(s,s')$, it will never have to be added again. Thus, the multiplicity of the parallel edges is at most two.

\begin{figure}[ht]
    \centering
    \includegraphics[page=7]{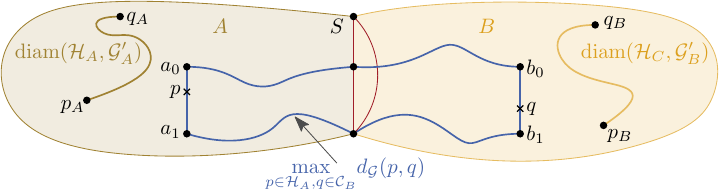}
    \caption{Visualization of the divide-and-conquer approach to compute $\diam(\GG)$ (see \cref{for:dc}).}
    \label{fig:divide_and_conquer_diam}
\end{figure}

Now we have two regimes depending on whether we want
to assume that the treewidth is constant, as done in~\cite{CABELLO2009815},
or whether we want to consider the treewidth a parameter, 
as done in~\cite{BHM20}.
The same distinction was made in~\cite{Cabello22}.
This difference affects the time to find a tree decomposition and
the number of portals in a balanced separation.
In both cases we use that a graph with $n$ vertices and treewidth $k$
has $O(kn)$ edges~\cite{Bodlaender98}. 

\begin{theorem}
\label{thm:diameter-treewidth-1}
	Let $k\ge 2$ be an integer constant, and let $\GG$ be the continuous graph defined by a graph $G$ with $n$ vertices, treewidth at most $k$, and
    nonnegative edge-lengths.
    Let $H$ be a subgraph of $G$ and let $\HH\subseteq \GG$
    be the corresponding continuous subgraph.
	The diameter $\diam(\HH,\GG)$ can be computed in 
    $O(n \log^{4k-2} n)$ time.
\end{theorem}

\begin{proof}
	If $G$ has fewer than $2k=O(1)$ vertices, 
    we compute $\diam(\HH,\GG)$ in $O(1)$ time. Otherwise, by \cref{theo:portals},
    we can find in linear time a separation $(A,B,S)$
    such that
    $|S|\le k$ and 
    $\frac{n}{k+1} \le |A| \le \frac{nk}{k+1}$.
    The pairs of portals $ss'\in E'$ to be added can be 
    computed using shortest-path trees in $G$, $G[A]$ and $G[B]$ from the vertices of $S$.
    This computation takes $O(k(m + n\log n)) =  O(n \log n)$ time.
    The edge sets defining $\HH_A$ and $\HH_C$ can be computed
    in $O(m)=O(n)$ time.
    
    Because the edge sets defining $\HH_A$ and $\HH_C$ are disjoint, 
    we can use \cref{lem:AB} to compute 
    the value $\max_{p\in \HH_A, q\in \HH_C} d_{\GG}(p,q)$ 
    in $O(m 2^{4k-3} k^4 B(m,4k-3))$ time.
    Using that $k$ is constant, $m=O(n)$, and~\cref{lem:Bnd},
    this time bound is $O(n \log^{4k-3}n)$.

    For $X\in\{A,B\}$, we construct the graphs $G'$ and $G'[X]$
    explicitly in $O(m)=O(n)$ time.
    Because adding the edges $ss'$  between all pairs of portals $ss'\in E'$ (with length $d_G(s,s')$) does not increase the treewidth (see \cref{theo:portals}) of the graph,
    the graphs $G'[A]$ and $G'[B]$ have treewidth at most $k$.
    The values $\diam(\HH_A,\GG'_A)$ and $\diam(\HH_C,\GG'_B)$ are 
    computed recursively, and we obtain $\diam(\HH,\GG)$ using
    \cref{for:dc}.

    Since $\frac{n}{k+1} \le |A| \le \frac{nk}{k+1}$ and $k$ is constant, 
    each side of the recursion has a constant fraction of the vertices 
    $|A|+|B|=n+k$,
    and the recursion depth is $O(\log n)$, leading to 
    a total running time of $O(n \log^{4k-2} n)$.
\end{proof}

\begin{theorem}
\label{thm:diameter-treewidth-2}
	Let $\GG$ be the continuous graph defined by a graph $G$ with $n$ vertices, treewidth at most $k$, and
    nonnegative edge-lengths.
    Let $H$ be a subgraph of $G$ and let $\HH\subseteq \GG$
    be the corresponding continuous subgraph.
	The diameter $\diam(\HH,\GG)$  can be computed in 
    $n^{1+\eps} 2^{O(k)}$ time, for any fixed $\eps>0$.
    \end{theorem}

\begin{proof}
	We use the same divide-and-conquer strategy as in \cref{thm:diameter-treewidth-1}. 
    The difference is in the properties of the separation.
    If $G$ has $O(k)$ vertices, we compute $\diam(\HH,\GG)$ in $O(k^2)$ time. 
    Otherwise, we proceed as follows.
    
	First, we note that, given a tree decomposition of $G$ of width $k'$, 
    we can obtain in linear time a separation 
    $(A,B,S)$ in $G$ with the following properties:
	the set $S$ of portals for $A$ has $k'+1$ portals;
    both $A$ and $B$ have $\Theta(n-k)$ vertices each;
    adding edges between the vertices of $S$ does not increase the treewidth
    \textit{of the tree decomposition}. 
    See for example~\cite[Theorem 19]{Bodlaender98};
    the set $S$ is a bag of the decomposition and thus the tree decomposition 
    keeps being valid with the addition of edges within $S$.

    It is shown in~\cite{BodlaenderDDFLP16} that,
	for graphs of treewidth at most $k$, one can find
	a tree decomposition of width $k'=3k+4$ in $2^{O(k)}n\log n$ time.
    From this we obtain the separation $(A,B,S)$ in $G$ mentioned above,
    where $|S|\le k'+1=3k+5$. 
    By \cref{lem:AB}, the value $\max_{p\in \HH_A, q\in \HH_C} d_{\GG}(p,q)$ 
    is computed in $O(m 2^{O(k)} (k'+1)^4 B(m,O(k))$ time.
    Using that $m=O(kn)$ and the estimate of \cref{lem:Bnd},
    this time bound becomes 
	\begin{align*}
		O((kn) 2^{O(k)} (k'+1)^4\cdot B(n,O(k)) =
		O(n2^{O(k)} \cdot n^{\eps}2^{O(k)}) =
		n^{1+\eps} 2^{O(k)},
	\end{align*}
	where $\eps>0$ can be chosen arbitrarily small.

    As done in the proof of \cref{thm:diameter-treewidth-1},
    we compute the edges between portals that are artificially added, that is, the set $E'$, as well as the values $d_G(s,s')$ for all $ss'\in E'$ in $O(k'(m+n\log n))=O(kn\log n)$ time.
    Then we compute $G'$, $G'[A]$, $G'[B]$, and the edge sets
    defining $\HH_A$ and $\HH_C$.
    
    To compute $\diam(\HH_X,\GG'_Y)$ recursively, where $(X,Y)$ is
    either $(A,A)$ or $(C,B)$,
    we pass to the subproblem the tree decomposition we have computed,
    trimmed to the vertices of $Y$ and adapted to keep the tree decompositions of size $O(|Y|)$.
    In this way, at any level of the recursion, we always have
    a set $S$ with $k'+1=3k+5$ portals. Thus, we compute the tree
    decomposition of width $k'$ only once, at the start, and then pass as part of the input a tree decomposition of width $k'$ to each subproblem.
	The recursive calls add a logarithmic factor to the total running time, 
    which is absorbed by the polynomial term $n^{1+\eps}$.

    (The reason for passing the tree decomposition to the subproblems
    is that adding the edges between the pairs of portals in $E'$ may give a clique of
    size $k'$, which increases the treewidth of $G'$. Computing an approximate
    tree decomposition of $G'$ anew would potentially increase the width of the decomposition
    at each level. However, if we pass the tree decomposition to the subproblems,
    we keep the width of the decomposition bounded by $k'=3k+4$ at all levels
    of the recursion.)
\end{proof}

\section{Mean distance in graphs with bounded treewidth}
\label{sec:meantreewidth}

In this section, we discuss the computation of the mean distance in continuous graphs with treewidth $k$.
The case of $k=1$ corresponds to continuous trees, one of the few graph classes for which linear-time algorithms are known~\cite{MeanDist23}.
Thus, we focus on $k>1$.
All our efforts will be on the computation of $\sumdist(\HH,\GG)$, from which $\mean(\HH,\GG)$ can be easily computed (both parameters are defined in \cref{subsec:parameters}.)
Again, we use orthogonal range searching, but now we need to efficiently retrieve sums of distances. To do this, we compute the sum of distances between pairs of edges as volumes of collections of triangular prisms (see also~\cite{MeanDist23}), and  represent them in a compact way. These representations can then be plugged in into the range searching machinery in order to obtain an algorithm as efficient as for the diameter.

\subsection{Mean distance between edges as a volume}\label{sebsec:mean edges}

A \emph{triangular prism} is a Cartesian product $\triangle\times [z_0,z_1]\subset \mathbb{R}^3$, where $\triangle$ is a triangle and $[z_0,z_1]$ is an interval; the lowest copy of the triangle, $\triangle\times \{z_0\}$, is its \emph{base}.
A \emph{truncated triangular prism} is the portion of a triangular prism $\triangle\times [z_0,z_1]$ in a lower, closed halfspace whose boundary separates $\triangle\times \{z_0\}$ from $\triangle\times \{z_1\}$. 
The \emph{heights} of a truncated triangular prism are the lengths
of the three edges orthogonal to the base.
The \emph{volume} of a truncated triangular prism with base $\triangle$ and
heights $h_1,h_2,h_3$ is $\tfrac{1}{3}\area(\triangle)(h_1+h_2+h_3)$.

Let $K=K(y,z,x_{0,0},x_{0,1},x_{1,0},x_{1,1})$ be a complete graph with 
vertex set $\{ a_0,a_1,b_0,b_1\}$ and variable edge lengths, as follows: 
\begin{itemize}
    \item[(i)] the edge $a_0a_1$ has variable length $y>0$,
    \item[(ii)] the edge $b_0b_1$ has variable length $z>0$,
    \item[(iii)] for all $\alpha,\beta\in \{ 0,1\}$, the edge 
        $a_\alpha b_\beta$  has length $x_{\alpha,\beta}\ge 0$.
\end{itemize}
Let $\mathcal{K}=\KK(y,z,x_{0,0},x_{0,1},x_{1,0},x_{1,1})$ denote the continuous graph defined by $K$.
See \cref{fig:K}.

We say that the $6$-tuple $(y,z,x_{0,0},x_{0,1},x_{1,0},x_{1,1})$
is \emph{compliant} if, for all $\alpha,\beta\in \{0,1\}$
it holds $x_{\alpha,\beta}= d_K(a_\alpha,b_\beta)$. Notice that the distance is taken in $K$.
In our setting we only need to consider compliant cases, which is enforced by conditions on the values that the variables can take.

 \begin{figure}
     \centering
     \includegraphics[page=1]{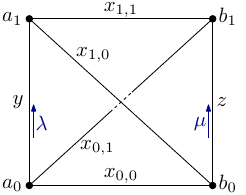}
     \caption{The graph $K(y,z,x_{0,0},x_{0,1},x_{1,0},x_{1,1})$.}
     \label{fig:K}
 \end{figure}

We want to understand how the total sum of distances between points on $a_0a_1$ and $b_0b_1$,
\[
    \xi(y,z,x_{0,0},x_{0,1},x_{1,0},x_{1,1}) = \iint_{p\in a_0a_1, q\in b_0b_1} d_{\KK}(p,q) \,dp \,dq,
\]
looks like.
For each $\lambda\in [0,y]$, let $p(\lambda)$ be the point 
specified by the triple $(a_0a_1,a_0,\lambda)$,
and, for each $\mu\in [0,z]$, let $q(\mu)$ be the point 
specified by the triple $(b_0b_1,b_0,\mu)$.
Then 
\[
    \xi(y,z,x_{0,0},x_{0,1},x_{1,0},x_{1,1}) = 
    \iint_{(\lambda,\mu)\in [0,y]\times [0,z]} d_{\KK}(p(\lambda),q(\mu)) 
        \,d\lambda \,d\mu .
\]
The function $(\lambda,\mu)\mapsto d_{\KK}(p(\lambda),q(\mu))$,
defined in $[0,y]\times [0,z]$,
is the lower envelope of the following four functions:
\[
    x_{0,0} + \lambda + \mu;~~
    x_{0,1} + \lambda + z-\mu;~~
    x_{1,0} + y-\lambda + \mu;~~
    x_{1,1} + y-\lambda + z-\mu. 
\]
Following~\cite{MeanDist23}, we call the graph of the function 
$(\lambda,\mu)\mapsto d_{\KK}(p(\lambda),q(\mu))$ a \emph{roof};
the value $\xi$ is the volume below the roof.
The minimization diagram of this function consists of convex pieces;
see \cref{fig:roof1}.
The gradient of each function 
is of the form $(\pm 1, \pm 1)$.
When the variable values are compliant, 
all four functions appear in the lower
envelope, and the minimization diagram has four regions.
(Some of them may be degenerate and contain only part of an edge of the domain.)

\begin{figure}[ht]
    \centering
    \includegraphics[page=5,width=\linewidth]{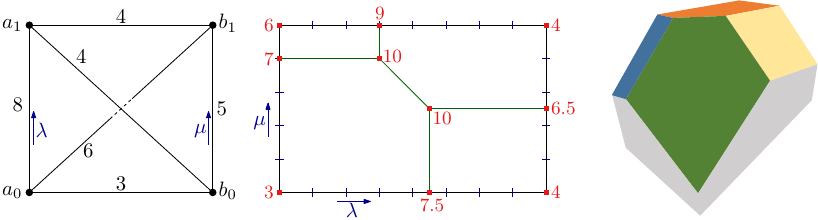}
    \caption{Concrete example of the minimization diagram of  $d_\KK(p(\lambda),q(\mu))$.
        In the center, some values of $d_\KK(p(\lambda),q(\mu))$
        are shown in red;
        on the right, 3D visualization of the roofs.}
    \label{fig:roof1}
\end{figure}

\begin{lemma}
\label{lem:polynomial}
    Consider the $4$-variable linear function 
    $L(x_{0,0},x_{0,1},x_{1,0},x_{1,1}) =    
    x_{1,0} + x_{0,1} - x_{0,0} - x_{1,1}$.
    There are two $6$-variable polynomials
    $\varrho_{+}(\cdot),\varrho_{-}(\cdot)$ of degree 
    at most three with the following property:
    when $(y,z,x_{0,0},x_{0,1},x_{1,0},x_{1,1})$ 
    is compliant,
    \[
        \xi(y,z,x_{0,0},x_{0,1},x_{1,0},x_{1,1}) = \begin{cases}
        \varrho_+(y,z,x_{0,0},x_{0,1},x_{1,0},x_{1,1}), & 
        \text {if } L(x_{0,0},x_{0,1},x_{1,0},x_{1,1})\ge 0,\\
        \varrho_-(y,z,x_{0,0},x_{0,1},x_{1,0},x_{1,1}), & 
        \text {otherwise}.
        \end{cases}
    \]    
\end{lemma}
\begin{proof}
    For the compliant case, there are two possible roofs, 
    as shown in \cref{fig:roof2}.
    We conclude that the following values exist:
    \begin{align*}
        \lambda_0\in [0,y] &\text{ such that } 
                d_K(a_0,b_0)+\lambda_0 = d_K(a_1,b_0)+ (y-\lambda_0),\\
        \lambda_1\in [0,y] &\text{ such that } 
                d_K(a_0,b_1)+\lambda_1 = d_K(a_1,b_1)+ (y-\lambda_1),\\
        \mu_0\in [0,z] &\text{ such that } 
                d_K(a_0,b_0)+\mu_0 = d_K(a_0,b_1)+ (z-\mu_0),\\
        \mu_1\in [0,z] &\text{ such that } 
                d_K(a_1,b_0)+\mu_1 = d_K(a_1,b_1)+ (z-\mu_1).
   \end{align*}
   Since in the compliant case we have $x_{\alpha,\beta}=d_K(a_\alpha,b_\beta)$, the solution is
   \begin{align*}
        \lambda_0 &=\tfrac{1}{2} (y + x_{1,0} - x_{0,0}), \\
        \lambda_1 &=\tfrac{1}{2} (y + x_{1,1} - x_{0,1}), \\
        \mu_0 &= \tfrac{1}{2} (z + x_{0,1} - x_{0,0}),\\
        \mu_1 & = \tfrac{1}{2} (z + x_{1,1} - x_{1,0}).
   \end{align*}
   To know in which case we are, we compare $\lambda_0$ and $\lambda_1$.
   Thus, it suffices to consider the sign of  
   \[
     (y + x_{1,0} - x_{0,0}) - (y + x_{1,1} - x_{0,1}) = 
        x_{1,0} + x_{0,1} - x_{0,0} - x_{1,1} =
    L(x_{0,0},x_{0,1},x_{1,0},x_{1,1}).
   \]
   Note that $\lambda_0\ge \lambda_1$ if and only if $\mu_0\ge \mu_1$.
   
   \begin{figure}
        \centering
        \includegraphics[page=3,width=\textwidth]{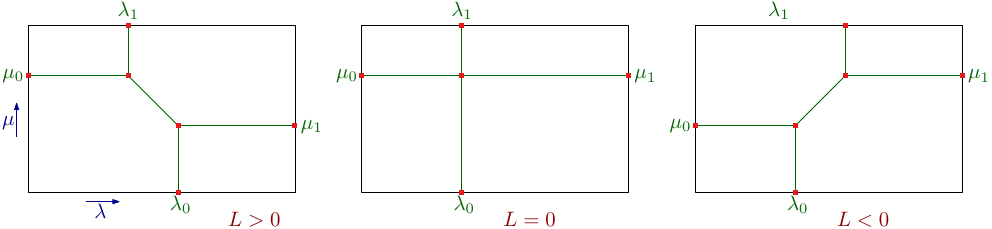}
        \caption{The cases where $L> 0$ (left), $L=0$ (center) and $L<0$ (right).}
        \label{fig:roof2}
    \end{figure}

    When we know how the minimization diagram looks like, we can
    compute the volume below the roof by breaking the minimization
    diagram into triangles. \cref{fig:roof3} shows the case when $L> 0$.
    The structure of the decomposition is the same, independently
    of the actual values of the parameters, as far as
$(y,z,x_{0,0},x_{0,1},x_{1,0},x_{1,1})$ is compliant and $L(x_{0,0},x_{0,1},x_{1,0},x_{1,1})>0$.

   \begin{figure}
        \centering
        \includegraphics[page=4]{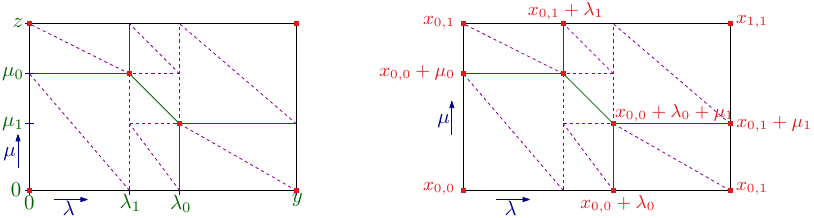}
        \caption{Breaking the domain into triangles to compute
            the volume for $L> 0$. Left: coordinates in the domain. Right: values
            $d_{\KK}(p(\lambda),q(\mu))$.}
        \label{fig:roof3}
    \end{figure}

    The area of each triangle is a linear
    function on $y,\lambda_0, \lambda_1$ multiplied
    by a linear function on $ z,\mu_0, \mu_1$.
    At the vertices of each piece, the value of
    $d_{\KK}(p(\lambda),q(\mu))$ is a linear 
    combination of $x_{0,0},y,\lambda_0, \lambda_1,z,\mu_0, \mu_1$.
    We conclude that 
    $\xi(y,z,x_{0,0},x_{0,1},x_{1,0},x_{1,1})$ is the sum
    of a finite number of cubic polynomials in the variables,
    and therefore it is itself a cubic polynomial.

    The case $L=0$ can be handled with the same formula;
    the area of six triangular regions of the decomposition vanish, but the formulas still hold.
    The case $L<0$ is handled similarly (recall \cref{fig:roof2}), but with a different polynomial.
\end{proof}

\subsection{Integral across the portals}
\label{sec:integral-across}
We reuse much of the notation and ideas from~\cref{sec:diam-across}. 
As before, $G$ is a graph with $n$ vertices and $m$ edges,
$(A,B,S)$ is a separation in $G$ with exactly $k$ portals,
and we fix an enumeration of the portals as $s_1,\dots,s_k$. 
We also fix an orientation for the edges of $G$, and consider sets $E_A\subseteq E(G[A])$  and $E_B\subseteq E(G[B])$, which are disjoint.
Our objective in this section is to compute the sum of distances between points in $E_A$ and $E_B$:
\[   \iint\limits_{{p\in \GG(E_A), q\in \GG(E_B)}} d_{\GG}(p,q)\, dp \, dq =
    \sum_{(a_0,a_1)\in E_A}~~\sum_{(b_0,b_1)\in E_B}
    ~~~~\iint\limits_{p\in a_0a_1, q\in b_0b_1} d_{\GG}(p,q)\, dp \, dq.
\]

For two fixed edges $(a_0,a_1)\in E_A$, $(b_0,b_1)\in E_B$, consider the continuous graph $\KK$ (see \cref{sebsec:mean edges}) defined by $K=K(\ell(a_0a_1), \ell(b_0b_1),
d_G(a_0,b_0), d_G(a_0,b_1), d_G(a_1,b_0), d_G(a_1,b_1))$.

Edges $a_0a_1$ and $b_0b_1$ belong to $\GG$ and $\KK$
and, for each $p\in a_0a_1$ and $q\in b_0b_1$
we have $d_\GG(p,q)=d_\KK(p,q)$.
Therefore,
\begin{align*}
    \iint\limits_{p\in a_0a_1, q\in b_0b_1} d_{\GG}(p,q) \,dp \,dq &=
    \iint\limits_{p\in a_0a_1, q\in b_0b_1} d_{\KK}(p,q) \,dp \,dq.
\end{align*}
It follows that our objective is to compute
\[ 
    \sum_{(a_0,a_1)\in E_A}~~\sum_{(b_0,b_1)\in E_B}
    \xi(\ell(a_0a_1),\ell(b_0b_1),d_G(a_0,b_0),
    d_G(a_0,b_1),d_G(a_1,b_0),d_G(a_1,b_1)).
\]

We keep using the predicates $\varphi(i;a,b)$
and $\Phi(\kappa; a_0,a_1,b_0,b_1)$ defined in~\cref{sec:diam-across},
where $a,a_0,a_1\in A$, $b,b_0,b_1\in B$,
$i\in [k]$ and $\kappa=(i_{0,0},i_{0,1},i_{1,0},i_{1,1})\in [k]^4$.

Consider the $4$-variable linear function 
    $L(x_{0,0},x_{0,1},x_{1,0},x_{1,1})$ of \cref{lem:polynomial}.
For each pair $(a_0,a_1),(b_0,b_1)\in E(G)$,
we have
\[
    (a_0,a_1,b_0,b_1) \text{ of \emph{type $1$}} ~~\Longleftrightarrow~~
    L(d_G(a_0,b_0), d_G(a_0,b_1),d_G(a_1,b_0),d_G(a_1,b_1))\ge 0.
\]
Otherwise, $(a_0,a_1,b_0,b_1)$ is of \emph{type $2$}.
For each oriented edge $(b_0,b_1)\in E(G)$ and type $\tau\in \{1,2\}$, we again consider the set $\Type_\tau(b_0,b_1)$, defined in \cref{subsec:walks}, \cref{eq:type12}. In addition, for
each $\kappa\in [k]^4$, each $(b_0,b_1)\in E_B$,
and each type $\tau\in \{ 1,2 \}$ we define
\begin{align}\label{eq:Psi1}
    \Psi_\tau(\kappa;b_0,b_1)= \sum \xi \big( \ell(a_0a_1), \ell(b_0b_1), 
        d_G(a_0,b_0), d_G(a_0,b_1),d_G(a_1,b_0),d_G(a_1,b_1) \big),
\end{align}
where the sum ranges over all oriented edges $(a_0,a_1)\in E_A$ such that
$(a_0,a_1)\in \Type_\tau(b_0,b_1)$ and $\Phi(\kappa;a_0,a_1,b_0,b_1)$ holds. 

Next, we show that we can efficiently compute  $\Psi_\tau(\kappa;b_0,b_1)$  
for all edges $(b_0,b_1)\in E_B$.

\begin{lemma}
\label{le:accross-mean}
    Consider a fixed type $\tau\in \{1,2\}$ and indices $\kappa\in [k]^4$.
    We can compute the values $\Psi_\tau(\kappa;b_0,b_1)$ for all $(b_0,b_1)\in E_B$ in $O(m 2^{4k-3} B(m,4k-3))$ time.
\end{lemma}
\begin{proof}
    We discuss the computation for $\tau=1$; the case $\tau=2$ is analogous
    and omitted. 
    
    Since the pair $(\GG,d)$ is a metric space, the triangle inequality tells us that $d_K(a_\alpha,b_\beta)=d_G(a_\alpha,b_\beta)$ for $\alpha,\beta\in \{0,1\}$. Hence, for each
    $(a_0,a_1)\in E_A$ and each $(b_0,b_1)\in E_B$, the tuple 
    $(\ell(a_0a_1),\ell(b_0b_1),
    d_G(a_0,b_0),d_G(a_0,b_1),
    d_G(a_1,b_0),d_G(a_1,b_1))$ is compliant.
    By \cref{lem:polynomial}, \cref{eq:Psi1} can then be
    written, for each $(b_0,b_1)\in E_B$, as  
    \begin{equation}\label{eq:Psi2}
        \Psi_1(\kappa;b_0,b_1) = 
        \sum \varrho_+ \big( \ell(a_0a_1), \ell(b_0b_1), 
        d_G(a_0,b_0), d_G(a_0,b_1),d_G(a_1,b_0),d_G(a_1,b_1) \big),
    \end{equation}
    where the sum ranges over all oriented edges $(a_0,a_1)\in E_A$ 
    such that $(a_0,a_1)\in \Type_1(b_0,b_1)$ and 
    $\Phi(\kappa;a_0,a_1,b_0,b_1)$ holds.
    
    Consider a fixed $\kappa=(i_{0,0},i_{0,1},i_{1,0},i_{1,1})\in [k]^4$.
    Again, we  omit  the dependency 
    on $\kappa$ when it is clear from the context.
    
    We reuse the construction in \cref{le:accross} that
    in $O(m 2^{4k-3} B(m,4k-3))$ time gives
    points $p(a_0,a_1)\in \mathbb{R}^{4k-4}$ (for each $(a_0,a_1)\in E_A$)
    and rectangles $R(b_0,b_1)\subset \mathbb{R}^{4k-4}$ such that
    \begin{align*}
        p(a_0,a_1)\in R(b_0,b_1)~~\Longleftrightarrow~~
        & \Phi(\kappa; a_0,a_1,b_0,b_1).
    \end{align*}
    For each edge $(a_0,a_1)\in E_A$, we append to $p(a_0,a_1)$
    an extra coordinate given by
    \[
        p_{\rm extra}(a_0,a_1)=d_G(a_1,s_{i_{1,0}}) + d_G(a_0,s_{i_{0,1}})
            - d_G(a_0,s_{i_{0,0}}) - d_G(a_1,s_{i_{1,1}})..
    \] 
    Let $p^+(a_0,a_1)$ be the resulting point in $\mathbb{R}^{4k-3}$.
    Similarly, we extend the rectangle $R(b_0,b_1)$ by an
    extra dimension by taking the Cartesian product with the interval
    \[
        I_{\rm extra}(b_0,b_1) = \big[d_G(s_{i_{0,0}},b_0) + d_G(s_{i_{1,1}},b_1)
    - d_G(s_{i_{1,0}},b_0) -  d_G(s_{i_{0,1}},b_1) , ~ +\infty \big).
    \]
    Let $R^+(b_0,b_1)$ be the resulting rectangle in $4k-3$ dimensions.
    Computing these points and intervals takes additional $O(m)$ time,
    once we have precomputed the distances from the portals to
    all vertices.

    We note that $p^+(a_0,a_1) \in R^+(b_0,b_1)$ if and only if 
    $\Phi(\kappa;a_0,a_1,b_0,b_1)$  
    and $p_{\rm extra}(a_0,a_1) \in I_{\rm extra}(b_0,b_1)$ holds. 
    The property $\Phi(\kappa;a_0,a_1,b_0,b_1)$
    translates into 
    \begin{quote}
        $\forall (\alpha,\beta)\in \{0,1\}^2$:~~ $i_{\alpha,\beta}$ is the smallest index $i$ such that $s_i$ lies in some shortest path
        from $a_\alpha$ to $b_\beta$. 
    \end{quote}
    The condition $p_{\rm extra}(\kappa;a_0,a_1)\in I_{\rm extra}(\kappa;b_0,b_1)$ is equivalent to 
    \begin{align*}
         d_G(a_1,s_{i_{1,0}}) + d_G(a_0,s_{i_{0,1}})
            - d_G(a_0,s_{i_{0,0}}) - d_G(a_1,s_{i_{1,1}}) \ge \\
     d_G(s_{i_{0,0}},b_0) + d_G(s_{i_{1,1}},b_1)
    - d_G(s_{i_{1,0}},b_0) -  d_G(s_{i_{0,1}},b_1),
    \end{align*}
    which is equivalent to
    \begin{align*}
         d_G(a_1,s_{i_{1,0}}) &+ 
         d_G(s_{i_{1,0}},b_0)
         + d_G(a_0,s_{i_{0,1}}) +d_G(s_{i_{0,1}},b_1)
            \ge \\
     & d_G(a_0,s_{i_{0,0}}) +d_G(s_{i_{0,0}},b_0) + d_G(a_1,s_{i_{1,1}})+d_G(s_{i_{1,1}},b_1).
    \end{align*}
    Under the assumption that 
    $\Phi(\kappa;a_0,a_1,b_0,b_1)$ holds, this condition becomes
    \[ 
        d_G(a_1,b_0) + d_G(a_0,b_1) - d_G(a_0,b_0) - d_G(a_1,b_1) \ge 0. 
    \]
    Since $L(x_{0,0}, x_{0,1}, x_{1,0}, x_{1,1})= 
        x_{1,0} + x_{0,1} - x_{0,0} - x_{1,1}$,
    this is precisely the condition for $(a_0,a_1,b_0,b_1)$ to be of type 1.
    We summarize:
    \[
        p^+(a_0,a_1) \in R^+(b_0,b_1) ~~\Longleftrightarrow ~~
            (a_0,a_1)\in \Type_1(b_0,b_1) \text{ and }
            \Phi(\kappa;a_0,a_1,b_0,b_1).
    \]
    Therefore, for each $(b_0,b_1)\in E_B$, we can rewrite \cref{eq:Psi2} as 
    \begin{equation}\label{eq:Psi3}
        \Psi_1(\kappa;b_0,b_1) = 
        \sum \varrho_+ \big( \ell(a_0a_1), \ell(b_0b_1), 
        d_G(a_0,b_0), d_G(a_0,b_1),d_G(a_1,b_0),d_G(a_1,b_1) \big),
    \end{equation}
    where the sum is over all points $p^+(a_0,a_1) \in R^+(b_0,b_1)$.

    For each $(a_0,a_1) \in E_A$,
    we compute and store the coefficients of the $5$-variable polynomial 
    \begin{align*}
        \varrho_{a_0,a_1}(z,&x'_{0,0},x'_{0,1},x'_{1,0},x'_{1,1}) =\\
        &\varrho_+\big( 
                \ell(a_0a_1),z,x'_{0,0}+\delta_{0,0}(a_0),
                x'_{0,1}+\delta_{0,1}(a_0),
                x'_{1,0}+\delta_{1,0}(a_1),
                x'_{1,1}+\delta_{1,1}(a_1) \big), 
    \end{align*}
    where we used the shorthand notation 
    $\delta_{\alpha,\beta}(a_\alpha)=d_G(a_\alpha,s_{i_{\alpha,\beta}})$.
    Note that this polynomial is cubic because it is essentially
    a shifted version of $\varrho_+(\cdot)$; one variable is also set
    to a constant value. Thus, we can compute and store 
    $\varrho_{a_0,a_1}(\cdot)$ in coefficient form in $O(1)$ time.
    
    We use orthogonal range searching to handle the sum in the right
    side of \cref{eq:Psi3}, as follows.
    We collect the set of points 
    $P_A=\{ p^+(a_0,a_1) \mid (a_0,a_1)\in E_A\}$.
    This set has $O(|E_A|)=O(m)$ points in $\mathbb{R}^{4k-3}$.
    For $P_A$ we compute $\mathcal{P}=\{ P_j \mid j\in J\}$ 
    and the data structure of \cref{thm:rangesearching}
    in $O(m k \cdot B(m,4k-3))$ time.
    
    For each $j\in J$, we compute and store the coefficients
    of the cubic, $5$-variable polynomial 
    \[
        \varrho_j (\cdot)= 
            \sum_{p^+(a_0,a_1) \in P_j} \varrho_{a_0,a_1}(\cdot).                 
    \]
    Computing and storing these polynomials in coefficient form
    over all $j\in J$ takes $O(\sum_j |P_j|)$ space, and does not 
    affect the asymptotic running time.
    This finishes the description of the construction of the data structure.

    Consider now an edge $(b_0,b_1)\in E_B$.
    We query the data structure with the rectangle $R=R^+(b_0,b_1)$
    and get the indices $J_R$ such that 
    $P_A\cap R^+(b_0,b_1) = \bigsqcup_{j\in J_R} P_j$.
    The query takes $O(2^{4k-3} B(m,4k-3))$ time
    and $|J_R|=O(2^{4k-3} B(m,4k-3))$.
    Next we compute in $O(|J_R|)=O(2^{4k-3} B(m,4k-3))$ time
    \begin{align*}
        \psi(b_0,b_1) = \sum_{j\in J_R}
            \varrho_j \big(\ell(b_0b_1),\epsilon_{0,0}(b_0),
                \epsilon_{0,1}(b_1),\epsilon_{1,0}(b_0), 
                \epsilon_{1,1}(b_1)\big),
    \end{align*}
    where we have used the shorthand notation
    $\epsilon_{\alpha,\beta}(b_\beta)=d_G(s_{i_{\alpha,\beta}},b_\beta)$.
    Note that the computation takes $O(1)$ time per $j\in J_R$
    because the polynomial $\varrho_j$ is of constant degree
    and the values where we evaluate are already available.
    
    We next argue that $\psi(b_0,b_1)=\Psi_0(\kappa;b_0,b_1)$.
    Using the definition of $\varrho_j$, we have
    \begin{align*}
        \psi(b_0,b_1) &= \sum_{j\in J_R} ~~\sum_{p^+(a_0,a_1)\in P_j} 
            \varrho_{a_0,a_1} \big(\ell(b_0b_1),\epsilon_{0,0}(b_0),
                 \epsilon_{0,1}(b_1),\epsilon_{1,0}(b_0), 
                 \epsilon_{1,1}(b_1)\big)\\
         &= \sum_{p^+(a_0,a_1)\in R^+(b_0,b_1)} 
            \varrho_{a_0,a_1} \big(\ell(b_0b_1),\epsilon_{0,0}(b_0),
                 \epsilon_{0,1}(b_1),\epsilon_{1,0}(b_0), 
                 \epsilon_{1,1}(b_1)\big).
    \end{align*}
    Because $p^+(a_0,a_1)\in R^+(b_0,b_1)$, the
    predicate $\Phi(\kappa;a_0,a_1,b_0,b_1)$ holds,
    and thus for all $\alpha,\beta\in \{0,1\}$
    \[
        \epsilon_{\alpha,\beta}(b_\delta)+\delta_{\alpha,\beta}(a_\alpha) = 
        d_G(s_{i_{\alpha,\beta}},b_\beta) + d_G(a_\alpha,s_{i_{\alpha,\beta}})
        = d_G(a_\alpha,b_\beta),
    \]
    which means that, whenever $p^+(a_0,a_1)\in R^+(b_0,b_1)$,
    we have
    \begin{align*}
        \varrho_{a_0,a_1} \big(\ell(b_0b_1),&\epsilon_{0,0}(b_0),
                 \epsilon_{0,1}(b_1),\epsilon_{1,0}(b_0), 
                 \epsilon_{1,1}(b_1)\big)=\\
        & \varrho_+ \big( \ell(a_0a_1), \ell(b_0b_1), 
        d_G(a_0,b_0), d_G(a_0,b_1),d_G(a_1,b_0),d_G(a_1,b_1) \big).
    \end{align*}
    Therefore, for the fixed $\kappa$ under consideration, 
    the value $\psi(b_0,b_1)$ we have computed is 
    \begin{align*}
        \sum_{p^+(a_0,a_1)\in R^+(b_0,b_1)}
                \varrho_+ \big( \ell(a_0a_1), \ell(b_0b_1), 
                   d_G(a_0,b_0), d_G(a_0,b_1),d_G(a_1,b_0),d_G(a_1,b_1) \big)
    \end{align*}
    which is $\Psi_1(\kappa;b_0,b_1)$ because of \cref{eq:Psi3}.    
    In total, we spend $O(2^{4k-3} B(m,4k-3))$ time per edge 
    $(b_0,b_1)\in E_B$ to compute $\Psi_0(\kappa;b_0,b_1)$.
\end{proof}

We can now efficiently compute the sum of distances between all pairs of points in, 
respectively, $\GG(E_A)$ and $\GG(E_B)$.

\begin{lemma}
\label{lem:AB-mean}
    If $E_A$ and $E_B$ are disjoint, we can compute 
    $\iint_{p\in \GG(E_A), q\in \GG(E_B)} d_{\GG}(p,q)\, dp \, dq$ 
    in $O(m 2^{4k-3} k^4 B(m,4k-3))$ time.
\end{lemma}
\begin{proof}
    We use \cref{le:accross-mean} to compute, for each type $\tau\in \{1,2\}$ and indices $\kappa\in [k]^4$, the values $\Psi_\tau(\kappa;b_0,b_1)$ for all $(b_0,b_1)\in E_B$.
    This takes $O(k^4)$ uses of \cref{le:accross-mean}, and therefore we spend 
    $O(m 2^{4k-3} k^4 B(m,4k-3))$ time.
    As discussed at the beginning of \cref{sec:integral-across}, we have that
    $\iint_{p\in \GG(E_A), q\in \GG(E_B)} d_{\GG}(p,q)\, dp \, dq$
    is 
    \[ 
        \sum_{(a_0,a_1)\in E_A}~~\sum_{(b_0,b_1)\in E_B}
        \xi(\ell(a_0a_1),\ell(b_0b_1),d_G(a_0,b_0),
        d_G(a_0,b_1),d_G(a_1,b_0),d_G(a_1,b_1)).
    \]
    Since for each $(a_0,a_1)\in E_A$ and each $(b_0,b_1)\in E_B$ there
    is exactly one $\kappa\in [k]^4$ such that $\Phi(\kappa;a_0,a_1,b_0,b_1)$,
    and because the types complement each other,
    this last sum is 
    \[
        \sum_{\kappa\in [k]^4}\sum_{\tau\in \{1,2\}} \sum_{(b_0,b_1)\in E_B} ~~ 
        \Psi_\tau(\kappa;b_0,b_1).
    \]
    With this expression, it is clear that we can obtain the result in
    $O(|E_B|\cdot k^4\cdot 2)$ additional time.
\end{proof}

\subsection{Global mean}
\label{sec:global-mean}
Let $G$ be a graph with $n$ vertices and $m$ edges defining
the continuous graph $\GG$. Let $H$ be a subgraph of $G$
defining a continuous graph $\HH\subseteq \GG$.
We use a divide-and-conquer approach to compute $\sumdist(\HH,\GG)$. 
The approach is very similar to that in \cref{sec:global-diameter}
for the diameter, and thus we only emphasize the differences.

Let $(A,B,S)$ be a separation in $G$. We use the definitions and notation introduced 
before \cref{thm:diameter-treewidth-1}: the graph $G'$ obtained after adding edges 
between portals to keep the distances in the subproblems, the continuous graphs $\GG_X$, 
$\HH_X$ and $\GG'_X$ for $X\in \{ A, B\}$, and the continuous graph $\HH_C \subseteq \HH_B$. Recall that for $X\in \{A,B\}$ and $p,q\in \GG_X$ we have $d_{\GG}(p,q)= d_{\GG'_X}(p,q)$.

Since $\HH=\HH_A\cup \HH_C$ and $\HH_A\cap \HH_C \subseteq S$, we have
\[ \HH\times \HH = (\HH_A\times \HH_A)\cup (\HH_A\times \HH_C)\cup (\HH_C\times \HH_A)\cup (\HH_C\times \HH_C)
\]
and the sets on the right side overlap in sets of measure zero.
Therefore
\begin{align*}\begin{aligned}\label{for:mean}
    \sumdist(\HH,\GG)&= ~~~
    2\cdot \iint\limits_{p\in \HH_A, q\in \HH_C} d_{\GG}(p,q)\, dp \, dq 
    + \sumdist(\HH_A,\GG)+\sumdist(\HH_C,\GG)\\
    &= ~~~
    2\cdot \iint\limits_{p\in \HH_A, q\in \HH_C} d_{\GG}(p,q)\, dp \, dq 
    + \sumdist(\HH_A,\GG'_A)+\sumdist(\HH_C,\GG'_B).
    \end{aligned}
\end{align*}

Because the edge sets defining $\HH_A$ and $\HH_C$ are disjoint, 
the first term can be computed using \cref{lem:AB-mean} in
$O(m 2^{4k-3} k^4 B(m,4k-3))$ time, where $k=|S|$.
The second and third terms can be computed recursively, and the size of the subproblems
is bounded by the size of the graphs $G[A]$ and $G[B]$.
The rest of the approach and analysis is essentially that of \cref{sec:global-diameter}.
Dividing $\sumdist(\HH,\GG)$ by $\ell(\HH)^2$, we obtain $\mean(\HH,\GG)$.

\begin{theorem}
\label{thm:mean-treewidth-1}
	Let $k\ge 2$ be an integer constant, and let $\GG$ be the continuous graph defined by a graph $G$ with $n$ vertices, treewidth at most $k$, and
    nonnegative edge-lengths.
    Let $H$ be a subgraph of $G$ and let $\HH\subseteq \GG$
    be the corresponding continuous subgraph.
	The mean distance
 $\mean(\HH,\GG)$ can be computed in 
    $O(n \log^{4k-2} n)$ time.
\end{theorem}

\begin{theorem}
\label{thm:mean-treewidth-2}
	Let $\GG$ be the continuous graph defined by a graph $G$ with $n$ vertices, treewidth at most $k$, and
    nonnegative edge-lengths.
    Let $H$ be a subgraph of $G$ and let $\HH\subseteq \GG$
    be the corresponding continuous subgraph.
	The mean distance $\mean(\HH,\GG)$ can be computed in 
    $n^{1+\eps} 2^{O(k)}$ time, for any fixed $\eps>0$.
\end{theorem}

\section{Planar graphs}
\label{sec:planar}
In this section, we consider a continuous graph $\GG$ defined by a planar graph $G$.
Note that, by Euler's formula, the number $F$ of faces of a planar graph is the same for any embedding of the graph,
and thus we can talk about \textit{the} number of faces of a planar graph.
Recall that the length $\ell(e)$ of an edge $e$ of $G$ is not related to the 
length of $e$ in any embedding or realization of $G$.

The main result of this section is the following:

\thmplan*

In order to prove the theorem, we fix a combinatorial embedding of $G$, obtaining a \emph{plane graph}.
A combinatorial embedding is usually described by giving the cyclic order of the edges around 
each vertex, and we will follow also this principle.
For any face $f$ of $G$, let $\GG_f$ be the continuous set of points on the boundary of $f$. 
Thus, $\GG_f$ contains precisely $\GG(e)$ over all edges $e$ defining the boundary of $f$.

Our main technical result is that, for each fixed face $f$ of $G$, one can compute the eccentricity 
and the mean distance of all points of $\GG_f$ in $O(n\log n)$ time. From this,
we immediately obtain \cref{thm:planar} by iterating through all faces.

In the following, we assume that the boundary of a face $f$ is a cycle. 
This can be achieved by cutting open some trees incident to $f$ and 
using zero-length edges. 

The rest of the section is organized as follows. First we introduce in \cref{sec:dynamicforests}
some data structures for dynamic trees. In \cref{sec:mssp} 
we explain how to maintain a shortest-path tree in $G$ as the source slides continuously 
along the boundary $\GG_f$ of a face $f$. We then address the computation of the eccentricity 
for points in $\GG_f$ in \cref{sec:eccentricity-plane}, while in \cref{sec:mean-plane}
we discuss the computation of the mean distance.

\subsection{Two data structures for dynamic forests}
\label{sec:dynamicforests}
We need two different data structures that dynamically maintain a forest with real weights. 

The first one is a \emph{vertex-weighted forest} 
that implicitly stores a real value with each vertex.
Moreover, each vertex has a boolean flag that tells whether
it is \emph{marked} or not; for some queries, only marked vertices are taken into account. 
The data structure supports the following operations:
\begin{itemize}
	\item $\Create(\lambda,b)$: Makes a new tree with a single vertex of 
		value $\lambda$. The vertex is marked or not depending on the boolean $b$.
	\item $\Cut(e)$: Removes the edge $e$ from the tree that contains it. 
	\item $\Link(u,v)$: Adds the edge $uv$.
	\item $\GetVertexValue(v)$: Returns the value of vertex $v$.
	\item $\AddTree(\Delta,v)$: Adds $\Delta$ to the  
        value of each vertex in the tree that contains $v$.
	\item $\MaxTree(v)$: Returns the maximum value
		over the \textit{marked vertices} in the tree that contains $v$.
    \item $\SumTree(v)$: Returns the sum of the values
		over all \textit{marked vertices} in the tree that contains $v$.
\end{itemize}

There are several data structures that can handle these operations in $O(\log n)$
amortized time per operation, such as Euler-tour trees~\cite{HenzingerK99,Tarjan97},
link-cut trees~\cite{GoldbergGT91,SleatorT83}, top trees~\cite{AlstrupHLT05}, 
or self-adjusting top trees~\cite{TarjanW05}.
Any of them will be suitable for our purposes.

The second data structure is an \emph{edge-weighted embedded forest} 
that implicitly stores a real value with each edge.
The forest is \emph{embedded}, meaning that the circular order of the edges around 
each vertex is prescribed. We consistently use \emph{clockwise order}
around each vertex for all vertices. 
We say that an edge $e$ is \emph{to the left} (resp.~\emph{to the right}) of a path 
$\pi=v_0v_1\ldots v_k$ if (see \cref{fig:trees1}, left):
\begin{itemize}
    \item $e$ is in the same connected component as $\pi$ but
        in not in $\pi$;
    \item the first edge $v_iu_i$ from $\pi$ towards $e$ (possibly $e=v_iu_i$) satisfies 
        that $i\neq 0$ and $i \neq k$; and 
    \item at $v_i$ the clockwise order restricted to $v_i u_i$ and 
        the two path edges incident to $v_i$ is $v_i v_{i-1}, v_i u_i, v_i v_{i+1}$ 
        (resp.,~the clockwise order is $v_i v_{i-1}, v_i v_{i+1}, v_i u_i$).
\end{itemize}
For trees, an edge cannot be both to the right and to the left of a path, because
it would introduce a cycle. By definition, if the first edge from the path $\pi$
towards the edge $e$ starts at an endpoint of $\pi$, then $e$
is neither to the left nor to the right of $\pi$.

\begin{figure}[ht]
\centering
	\includegraphics[page=1]{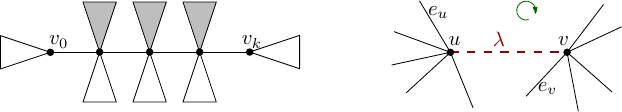}
     \caption{Left: edges in the shadow subtrees are to the left of the path from $v_0$ to $v_k$.
        Right: operation $\Link(u,v,e_u,e_v,\lambda_{uv})$ for 
		edge-weighted embedded forest.}
	\label{fig:trees1}
\end{figure}

Similarly to the previous scenario, each edge has a boolean flag that tells 
whether it is \emph{marked} or not; 
for some queries, only marked edges are taken into account. 
The data structure for edge-weighted embedded forests supports 
the following operations. 
\begin{itemize}
	\item $\Create()$: Makes a new tree with a single vertex.
	\item $\Cut(e)$: Removes the edge $e$ from the tree that contains 
        it. 
	\item $\Link(u,v,e_u,e_v,\lambda,b)$: Adds the edge $uv$ 
		with value $\lambda$; it is marked or not depending on the boolean $b$.
		The edge $uv$ is inserted such that it is clockwise after $e_u$ at $u$ 
		and clockwise after $e_v$ at $v$; see \cref{fig:trees1}, right.
        If $u$ (resp.~$v$) had degree $0$ before
        the edge insertion, the value $e_u$ (resp.~$e_v$) is empty.
	\item $\GetEdgeValue(e)$: Returns the value of edge $e$.
	\item $\AddLeftPath(\Delta,u,v)$: Adds the value $\Delta$ to each edge
		to the left of the path from $u$ to $v$; 
		recall \cref{fig:trees1}, left.
	\item $\MaxTree(u)$: Returns the  
        maximum value over the \textit{marked edges} in the tree that contains $u$.
\end{itemize}

Note that $\AddLeftPath(\Delta,u,v)$ and $\AddLeftPath(\Delta,v,u)$ give
different results. 
Calling $\AddLeftPath(\Delta,v,u)$ adds $\Delta$ to the edges to the {\em right} of the path from $u$ to $v$.

We have not been able to trace in the literature a data structure 
with an operation like $\AddLeftPath$. 
There are at least two data structures that can handle these operations in $O(\log n)$
amortized time per operation.
The first option is to use self-adjusting top trees~\cite{TarjanW05}, 
which were explicitly designed to keep track of the cyclic order of the edges
around each vertex. 
The second option is to use an adaptation of top trees~\cite{AlstrupHLT05} 
to handle embeddings, as used by Holm and Rotenberg~\cite{HolmR17} (without weights).

\begin{lemma}\label{lem:embeddedtree}
	An edge-weighted embedded forest with marked edges and 
    operations $\Create$, $\Cut$, $\Link$, $\GetEdgeValue$,
	$\AddLeftPath$ and $\MaxTree$ can be maintained in $O(\log n)$  amortized time per operation,
	where $n$ is the number of vertices in the data structure.
\end{lemma}

Since the proof of \cref{lem:embeddedtree} is long, technical, and does not relate to the 
main thread of our work, we have placed it in \cref{app:proof-data-structure}.

\subsection{Multiple-source shortest paths}
\label{sec:mssp}
We build upon previous   algorithms for multiple-source shortest-paths 
for planar graphs%
~\cite{CabelloCE13,DasKGW22,EricksonFL18,Klein05}.
In this setting, we have a fixed face $f$ in a plane graph, 
and want to encode all shortest-path trees
from {\em each} vertex incident to $f$.
For our problem in a continuous setting, it seems most convenient to follow the paradigm of 
Cabello, Chambers and Erickson~\cite{CabelloCE13}, because it is already
based on sliding the source continuously along an edge.

Let $s$ (for \emph{source}) be a vertex  
on an edge $uv$ of $f$ and let $T_s$ 
be a shortest-path tree rooted at $s$ to all vertices of $G$. 
We want to maintain $T_s$ as $s$ slides 
continuously along $uv$, say from $u$ to $v$.
In general, $T_s$ consists of the edges $su$ and $sv$, 
a subtree $T_u$ rooted at $u$, which we call the \emph{red subtree}, and 
a subtree $T_v$ rooted at $v$, which we call the \emph{blue subtree}.
It is possible that the edge $su$ does not belong to $T_s$, and in this case the red subtree is empty; similarly, if the edge $sv$ does not belong to $T_s$, the blue subtree is empty.
We further classify the edges of $G$ into three groups: 
the \emph{red edges} have both endpoints in the red tree $T_u$,
the \emph{blue edges} have both endpoints in the blue tree $T_v$,
and \emph{green edges} have one endpoint in each subtree. See \cref{fig:red-blue}.

\begin{figure}
\centering
	\includegraphics[page=1,width=\textwidth]{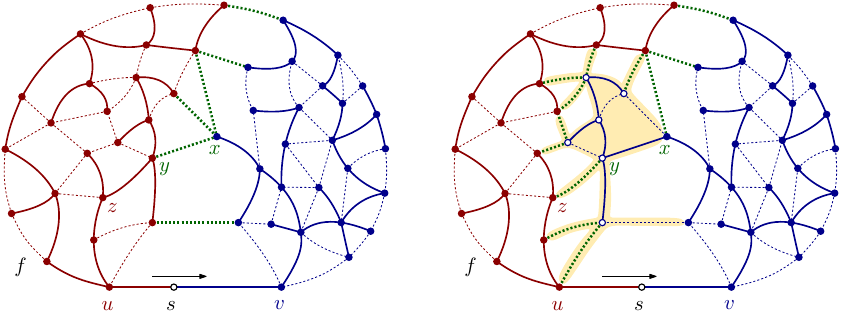}
     \caption{The red subtree, the blue subtree, and the green edges.
		Edges that do not belong to $T_s$ are dotted.
		From left to right we are making a pivot operation in $T_s$ 
		where $xy$ pivots in and $zy$ pivot out. The red
		vertices that become blue with the pivot operation are marked
		with white interior. The yellow region tells where changes occur.}
	\label{fig:red-blue}
\end{figure}

In a generic situation, as we slide $s$ towards $v$ by a {\em sufficiently small}
length $\lambda>0$, we increase the distance from $s$ to each red vertex by
$\lambda$, and decrease the distance from $s$ to each blue vertex by $\lambda$.
During the sliding of $s$, at some critical points the red and blue trees 
change and have to be updated. 
More precisely, consider a green edge $xy$ such that $x$ is blue and $y$ is red.
If at some moment the source $s$ reaches a point such that 
$d_{T_s}(s,v)+d_{T_s}(v,x)+\ell(xy)=d_{T_s}(s,u)+d_{T_s}(u,y)$,
then the edge $xy$ enters the blue tree $T_v$, and the edge connecting $y$ to 
its parent $z$ in the red tree $T_u$ has to be removed. With this operation,
the subtree of $T_u$ rooted at $y$ becomes blue, and the color of some edges
changes.
This is called a \emph{pivot operation}, where the edge $xy$ pivots in $T_s$ and 
the edge $zy$ pivots out. See \cref{fig:red-blue}.
(In non-generic cases, one has to be more careful about when an edge pivots in.)

So far, we have not used planarity in any way, and the source $s$ can slide 
along any edge. If we  restrict ourselves to {\em plane graphs} and
the edge lengths are \textit{generic} (meaning that all vertex-to-vertex shortest paths are unique) 
Cabello, Chambers and Erickson~\cite{CabelloCE13} show that
each edge pivots in and out $O(1)$ times, and at each moment there is only
one candidate pivot to perform. They also show how to remove the generic edge weight assumption by applying a randomized perturbation
scheme that works with high probability~\cite{CabelloCE13}.
Alternatively, 
Erickson, Fox and Lkhamsuren~\cite{EricksonFL18} present a {\em deterministic} 
perturbation scheme that guarantees genericity and has no asymptotic overhead. 
Note that the perturbations are made only to detect
in which order the pivot operations have to be performed, but the operations
can then be performed in the unperturbed, original setting.

To describe the places where changes occur when $s$ slides along $\GG_f$, we need a cyclic order of the points of $\GG_f$. 
For two points $s_0,s_i$ along $\GG_f$, 
the \emph{$f$-interval} $[s_0,s_i]_f$ is the set of points $\tilde s\in \GG_f$  such that the counterclockwise traversal of $f$ starting from $s_0$ passes through $\tilde s$ before passing $s_i$. 
See \cref{fig:f-intervals}, also for easier parsing of \cref{thm:pivoting}.

\begin{figure}
\centering
	\includegraphics[page=3]{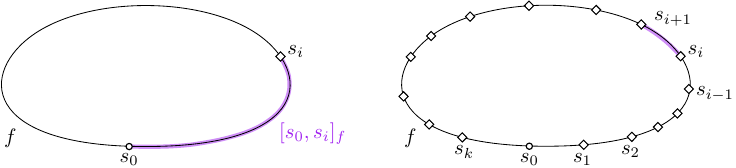}
     \caption{Left: example of $f$-interval.
		Right: example showing the cyclic order of $s_0,s_1,\dots,s_k$ 
		in \cref{thm:pivoting}. Note that, in general, $s_i$ 
		is not a vertex, but a point on $\GG_f$.
		For all sources $s$ in the $f$-interval $[s_i,s_{i+1}]_f$, 
		$T_s$ is the tree $T_i$.}
	\label{fig:f-intervals}
\end{figure}

\begin{theorem}[Cabello, Chambers and Erickson~\cite{CabelloCE13} 
combined with Erickson, Fox and Lkhamsuren~\cite{EricksonFL18}]
\label{thm:pivoting}
	Let $G=(V,E)$ be a plane graph with $n$ vertices and non-negative edge lengths.
	Let $f$ be any fixed face of $G$ such that the facial walk of $f$ is a cycle, 
	let $s_0$ be any vertex on the boundary of $f$ 
	and let $T_0$ be the shortest-path tree rooted at $s_0$.
	A sequence 
	of triples $(s_1,e_1,e'_1),\dots,(s_k,e_k,e'_k)\in \GG_f\times E\times E$ can be computed in $O(n\log n)$ time such that:
	\begin{itemize}
	\item $k=O(n)$;
	\item $s_i\in [s_0,s_{i+1}]$ for all $i\in [k-1]$;
	\item $T_0$ is a shortest-path tree for all sources $s\in [s_0,s_1]_f$;
	\item the tree $T_i$ (for $i\in [k]$),
		obtained inductively from $T_{i-1}$ by pivoting $e_i$ in and pivoting $e'_i$ out,
		is a shortest-path tree for all sources $s\in [s_i,s_{i+1}]_f$ \textup(where $s_{k+1}=s_0$\textup).
	\end{itemize}
\end{theorem}

\subsection{Eccentricity from a face}
\label{sec:eccentricity-plane}

We keep assuming that the boundary of $f$ is a cycle.
Consider any source $s$ on $\GG_f$ and let $T_s$ be the shortest-path tree from $s$. 
Let $C_s=E(G)\setminus E(T_s)$, that is, the set of edges 
not in the shortest-path tree $T_s$.
For each edge $xy\in C_s$, there is one point $m_s(xy)$ on edge $xy\subset \GG$
with two shortest paths from $s$, one through $x$ and one through $y$.
See \cref{fig:slides}.
For each edge $xy\in C_s$, we define the weight $\omega_s(xy)$ as
the distance from $s$ to $m_s(xy)$, thus
$\omega_s(xy) = d_{\GG}(s,m_s(xy)) = (d_{T_s}(s,x)+ \ell(xy) + d_{T_s}(s,y))/2$.

Let $H$ be a subgraph of $G$ defining the continuous subgraph $\HH$
of $\GG$. 
Recall that the eccentricity of point $p\in \GG$ with respect to $\HH$ is 
$\ecc(p,\HH,\GG)=\max_{q \in \HH} d_{\GG}(p,q)$.
We have the following characterization.

\begin{figure}
	\centering
		\includegraphics[page=4,width=\textwidth]{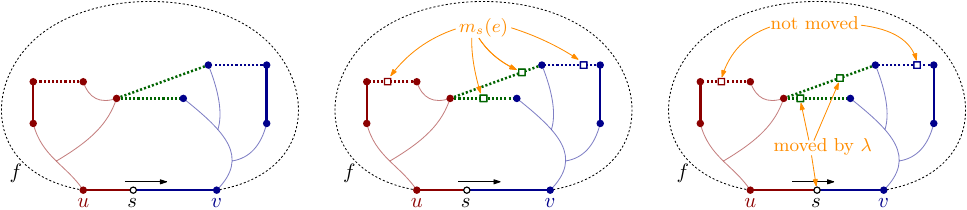}
		 \caption{The vertices $m_s(e)$ for $e\in C_s$ and how they move.}
		\label{fig:slides}
\end{figure}
	
\begin{lemma}
\label{lem:ecc-onestep}
	For each $s\in \GG_f$, the eccentricity of $s$ in $\GG$ is 
	\[
		\ecc(s,\HH,\GG) ~=~ \max\left\{ \max_{x\in V(H)} d_{T_s}(s,x) ,~ 
			\max_{xy\in C_s\cap E(H)} \omega_s(xy) \right\}.
	\]
\end{lemma}
\begin{proof}
	On the right hand side we only have distances from $s$ to points
    on $\HH$, and therefore	the eccentricity of $s$ 
    with respect to $\HH$ is at least as large as the right hand side.
	To see the other direction, consider the point $p^*\in \HH$ such
	that $d_\GG(s,p^*)=\ecc(s,\HH,\GG)$.
	It cannot be that $p^*$ lies in the interior of an edge $e$ 
    of $T_s\cap H$
	because at least one of the endpoints of $e$ is further from $s$.
	Therefore, $p^*$ lies on a vertex of $H$ or in the interior of an
	edge of $E(H)\setminus E(T_s)= C_s\cap E(H)$.
	If $p^*$ is a vertex of $H$, then $d_\GG(s,p^*)=d_{T_s}(s,p^*)$
	is considered on the right side.
	If $p^*$ lies on the interior of the edge $xy$ of $C_s\cap E(H)$,
	then $d_\GG(s,p^*)\le d_{\GG}(s,m_s(xy))= \omega_s(xy)$,
	which also appears on the right side.
\end{proof}

Consider the dual {\em plane} graph $G^*$ of $G$. Thus, $G^*$ is
another embedded (multi-)graph where vertices correspond to faces of $G$, 
and vice versa. Two vertices of $G^*$ are connected by an edge
for each edge that is shared by their corresponding faces of $G$; 
therefore, every edge $e$ in $G$ has a corresponding 
dual edge $e^*$ in $G^*$. Similarly, every vertex $v$ of $G$
corresponds to a face $v^*$ of $G^*$, and every face $f$ of $G$ 
corresponds to a vertex $f^*$ of $G^*$. Moreover, each dual 
edge $e^*$ is embedded so that it crosses the corresponding primal edge $e$ exactly once, and intersects no other primal edge. This fixes the rotation of the
dual edges incident to a vertex $f^*$ of $G^*$. 

For the set $C_s$, define $C^*_s=\{e^*\mid e\in C_s \}$.
An important insight noted in previous work is that $C^*_s$ defines a spanning
tree of the dual graph $G^*$. See \cref{fig:dual}.
The pair $(E(T_s),C_s)$ forms a so-called
\emph{tree-cotree decomposition} of $G$~\cite{Eppstein03}.
This means that we can use a data structure to store 
and manipulate $C_s$ as an {\em embedded} tree in the dual. 
We add to $C_s^*$ two artificial vertices as follows: 
let $g$ be the other face incident to the edge $uv$ that contains $s$;
we add a vertex $a$ and the edge $af^*$ such that, cyclically around $f^*$, 
its order is adjacent to $f^*g^*$; we add a vertex $b$ and
the edge $bg^*$ such that, cyclically around $g^*$,
its order is that of $f^*g^*$.
Note that an edge $e$ of $C_s$ is green if and only if $e^*$ 
lies in the $a$-to-$b$ path, it is red if and only if $e^*$ lies in a subtree to the right of the $a$-to-$b$ path, and it is blue 
if and only if $e^*$ lies in a subtree to the left of the $a$-to-$b$ path in $C^*_s$.

\begin{figure}
\centering
	\includegraphics[page=2, width=\textwidth]{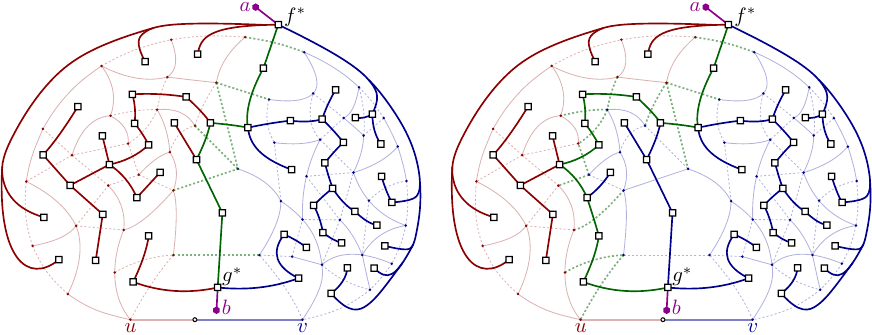}
     \caption{The dual tree $C_s^*$ for the examples in  \cref{fig:red-blue}. The vertices $a$ and $b$ are added artificially as segments of the dual edge $f^*g^*$.}
	\label{fig:dual}
\end{figure}

\begin{theorem}
    \label{thm:oneface}
    Let $\GG$ be the continuous graph defined by a plane graph $G$ 
    with $n$ vertices and nonnegative edge-lengths.
    Let $f$ be a face of $G$ and let $\GG_f$ be the set of points of $\GG$ on the boundary of $f$.
	Let $H$ be a subgraph of $G$ and let $\HH\subseteq \GG$ be the corresponding continuous subgraph.
    The maximum eccentricity $\ecc(s,\HH,\GG)$ over all points 
    $s$ in $\GG_f\cap \HH$ can be computed in $O(n\log n)$ time.
\end{theorem}
\begin{proof}
	We change the embedding of $G$, if needed, to assume that $f$
	is the outer face of $G$. 
	As mentioned before, we may assume that the boundary of $f$ is a cycle.
	We traverse the boundary of the face $f$ counterclockwise.
    We spend $O(n)$ time to make these changes, if needed.
     
	We compute and store the embedded dual graph $G^*$ of $G$.
	We maintain pointers between the edges in the dual and the primal.
	Since the embedding is going to be static,
	for each vertex of $G$ or $G^*$, we store explicitly 
	the cyclic order of the edges incident to it and rank them so
	that, given a sequence of three edges incident to the same vertex, we can
	identify in $O(1)$ time whether they are clockwise or counterclockwise. 
    Again, this takes linear time.
	
	For a fixed vertex $s_0$ from $f$, we compute the shortest-path tree $T_0=T_{s_0}$
	from $s_0$ to the vertices of $G$. 
	Computing the sequence $(s_1,e_1,e'_1),\dots,(s_k,e_k,e'_k)\in \GG_f\times E\times E$
	of \cref{thm:pivoting} takes $O(n \log n)$ time.
	
	For $i\in [k]$, we write $d_i$, $T_i$, $C_i$, $C_i^*$, $m_i$ and $\omega_i$
	instead of $d_{s_i}$, $T_{s_i}$, $C_{s_i}$, $C_{s_i}^*$, $m_{s_i}$ and $\omega_{s_i}$.

	We store $T_0$ using the data structure for vertex-weighted graphs of \cref{sec:dynamicforests} where,
	for each vertex $u$ of $G$, the associated value is $d_{T_0}(s_0,u)$
    and the marked vertices are those of $V(H)$.
	We construct the embedded dual tree $C^*_0$ and store it using 
	the edge-weighted embedded forest of \cref{lem:embeddedtree},
	where, for all $xy\in C_0$, the value associated with the edge $(xy)^*$ is 
	$\omega_{0}(xy)$. Here, the marked edges are those dual to $C_0\cap E(H)$. 
    This initialization of $T_0$ and $C^*_0$ takes $O(n\log n)$ time.
	
	Now, we explain how to construct $T_i$ and $C_i^*$ for any $i\in [k]$, 
    given a representation of $T_{i-1}$,
	where each vertex $x$ of $G$ has label $d_{i-1}(s_{i-1},x)$,
	and a representation of the embedded dual tree $C_{i-1}^*$ where 
	each edge $xy\in C_{i-1}$ has value $\omega_{i-1}(xy)$. 
	Let $\lambda_i$ be the length of the $f$-interval 
	$[s_{i-1},s_i]_f$. As we move the source from $s_{i-1}$ to $s_i$
    along an edge of $f$, 
    we note the following (see \cref{fig:red-blue,fig:slides}):
    
	\begin{itemize}
		\item For each vertex in the red subtree of $T_{i-1}$, the distance in $T_i$ increases 
            by $\lambda_i$.  For each vertex in the blue subtree of $T_{i-1}$, its 
            distance in $T_i$ decreases by $\lambda_i$.
		\item For each green edge $xy$, the distance from $s_{i-1}$ to $m_{i-1}(xy)$ 
            is the same as the distance from $s_{i}$ to $m_i(xy)$.
			Thus, $\omega_i(xy)=\omega_{i-1}(xy)$ for green edges $xy$.
		\item For each red edge $xy$ of $C_{i-1}$, the point $m_{i-1}(xy)$ is equal to $m_i(xy)$,
			but the distance from $s_{i}$ to $m_i(xy)$ increases by $\lambda_i$ 
            with respect to the distance from $s_{i-1}$ to $m_{i-1}(xy)$.
		\item For each blue edge $xy$ of $C_{i-1}$, the point $m_{i-1}(xy)$ is equal to $m_i(xy)$,
			but the distance from $s_{i}$ to $m_i(xy)$ decreases by $\lambda_i$ with respect to the distance from $s_{i-1}$ to $m_{i-1}(xy)$.
	\end{itemize}
	Therefore, we obtain $T_i$ from $T_{i-1}$ performing the following
	operations in the dynamic tree: if $s_{i-1}u$ belongs to $T_{i-1}$,
	we perform $\Cut(s_{i-1}u)$, $\AddTree(\lambda_i,u)$ and $\Link(s_{i-1},u)$; 
	if $s_{i-1}v$ belongs to $T_{i-1}$, we perform 
	$\Cut(s_{i-1}v)$, $\AddTree(-\lambda_i,v)$ and $\Link(s_{i-1},v)$;
	then we perform the pivot operation of cutting $e_i$ and linking $e'_i$.
    The set of marked vertices ($V(H)$) remains unchanged.
    
	To obtain $C^*_i$ from $C^*_{i-1}$, we perform the following
	operations in the dynamic tree storing $C^*_{i-1}$.
	Since the blue edges of $C_s$ are precisely those to the left of
	the path from $a$ to $b$, we perform $\AddLeftPath(-\lambda_i,a,b)$.
	Since the red edges of $C_s$ are precisely those to the right of
	the path from $a$ to $b$, we perform $\AddLeftPath(\lambda_i,b,a)$.
	The value of the green edges did not change.
	Now, we have to perform the pivot operation of deleting $(e'_i)^*$ 
	and inserting $e^*_i=(xy)^*$ in $C_i$; note the reverse order
	of operations with respect to the primal tree.
	For this, we $\Cut(e'_i)$, use the data structure for the primal $T_i$
    to obtain $\GetVertexValue(x)=d_i(s_i,x)$ and $\GetVertexValue(y)=d_i(s_i,y)$,
	and compute $\omega_i(xy)$.
	Using the precomputed circular order of $(xy)^*$, we identify
	where to insert the dual edge $(xy)^*$, and call $\Link$ to insert it.
    If $xy\in E(H)$, then the dual edge $(xy)^*$ is inserted as marked,
    and otherwise it is inserted as unmarked.
	
	This finishes the description of how to obtain $T_i$ and $C^*_i$
	from $T_{i-1}$ and $C^*_{i-1}$. At the end of each iteration,
	we can compute $\max_{x\in V(H)} d_{T_i}(s_i,x)$ by performing
	$\MaxTree(s_i)$ in the data structure storing $T_i$,
	and we can compute $\max_{e\in C_i\cap E(H)} \omega_i(e)$ using 
	$\MaxTree(f^*)$ in the data structure storing $C^*_i$.
	From this, due to \cref{lem:ecc-onestep}, we obtain $\ecc(s_i,\HH,\GG)$.

    When $s$ slides over a vertex of the face, we also need to 
    make a small update and the concept of red and blue has
    to be updated; we skip the easy details.
 
	At each $i\in [k]$, we spend $O(\log n)$ amortized time for $O(1)$
	operations in dynamic forests and $O(1)$ additional work. 
	Thus, the eccentricity of each $s_i$ with respect to $\HH$ 
    is computed in $O(n\log n)$ time. Finally, note that the eccentricity 
    has to be attained either at $s_0$ or at some point where the shortest-path 
    tree changes. Therefore 
    \[
        \max_{s\in \GG_f\cap \HH}\ecc(s,\HH,\GG)=
        \max\{ \ecc(s_i,\HH,\GG)\mid i\in [k] \text{ with } s_i\in \HH\},
    \]
    and the result follows.
\end{proof}

Applying \cref{thm:oneface} to each face of an embedding of an $n$-vertex graph $\GG$ 
with $F$ faces, we can compute the diameter $\diam(\HH,\GG)$ in $O(nF\log n)$ time 
(compare to \cref{thm:planar}).

\subsection{Mean distance from a face}
\label{sec:mean-plane}

The approach used to compute the diameter of planar graphs can be easily adapted to 
compute $\int_{s\in \GG_f\cap \HH} \int_{q\in \HH} d_{\GG}(s,q) \,dq \,ds$,
as follows.

As we move the source $s$ along $\GG_f$, with each edge $e$ we maintain 
the value 
$\nu(s,e)=\int_{q\in e} d_\GG(s,q) \, dq$.
When we slide $s$ by $\lambda>0$ \textit{without pivoting} into a new source $s'$,
the new value $\nu(s',e)$ satisfies
\[
   \nu(s',e) = \begin{cases}
                \nu(s,e)+\lambda\cdot \ell(e)&\text{ if $e$ is red,}\\
                \nu(s,e)-\lambda\cdot \ell(e)&\text{ if $e$ is blue,}\\
                \nu(s,e)                &\text{ if $e$ is green.}
                \end{cases}
\]

Consider the situation where we slide the source $s$ along the $f$-interval 
$I_i:= [s_{i-1},s_i]_f$ and let $\lambda_i$ be the length of this interval $I_i$. 
No pivoting occurs during the slide.
Then, for each red edge $e$ of $\GG$ we have 
\[
    \int_{s\in I_i}\int_{q\in e} d_\GG(s,q)  \, dq \, ds ~=~
        \int_{\lambda'\in [0,\lambda_i]} \int_{q\in e} 
        \big( d_\GG(s,q)+\lambda'\big)  \, dq \, d\lambda' ~=~
    \lambda_i \cdot \nu(s_{i-1},e) + \ell(e)\cdot \lambda_i^2/2,
\]
for each blue edge $e$ we have
\[
    \int_{s\in I_i}\int_{q\in e} d_\GG(s,q)  \, dq \, ds ~=~
    \int_{\lambda'\in [0,\lambda_i]} \int_{q\in e} 
        \big( d_\GG(s,q)-\lambda'\big) \, dq \, d\lambda' ~=~
    \lambda_i \cdot \nu(s_{i-1},e) -  \ell(e)\cdot \lambda_i^2/2,
\]
and for each green edge $e$ we have
\[
    \int_{s\in I_i}\int_{q\in e} d_\GG(s,q) \, dq \, ds ~=~
    \int_{\lambda'\in [0,\lambda_i]} \int_{q\in e} 
        d_\GG(s,q) \, dp \, d\lambda' =
    \lambda_i \cdot \nu(s_{i-1},e). 
\]
The edge $uv$ containing $I_i$ has to be treated as a special case,
and has a closed formula depending on $d_\GG(s_i,u)$, $d_\GG(s_i,v)$,
and $\lambda_i$. See \cref{sebsec:mean edges}.

Since each edge of $H-uv$ is either red, blue, or green,
we have
\begin{align*}
    \int_{s\in I_i}\int_{q\in \HH\setminus uv} &d_\GG(s,q)  \, dq \, ds
\\ &\hspace{-.5cm}=~
    \lambda_i \cdot \sum_{e\in E(H)\setminus \{uv\}}\nu(s_{i-1},e) + 
    (\lambda_i^2/2) \cdot \Big(\sum_{e\in E(H)\text{ red}}\ell(e) - 
                            \sum_{e\in E(H) \text{ blue}}\ell(e)\Big).
\end{align*}
It follows that to compute 
\[
    \int_{s\in I_i}\int_{q\in \HH} d_\GG(s,q)  \, dq \, ds
\]
it suffices to compute, for each $s\in \{s_0,\dots, s_k\}$, the values
\begin{equation}\label{eq:mean-planar}
    \sum_{e\in E(H)} \nu(s,e),~~
    \sum_{e\in E(H)\text{ red}}\ell(e),~~ 
    \sum_{e\in E(H) \text{ blue}}\ell(e).
\end{equation}
(Note that the concept of red and blue edges depends on $s$.)

As for the case of eccentricity, these values can be maintained 
using the dynamic data structures of \cref{sec:dynamicforests}.
For this, we use the decomposition of $E(G)$ into $T_s$ and $C_s$.

For the edges of $T_s$, we use a dynamic forest with marked vertices
and vertex weights. A simple option is to subdivide each
edge $e\in E(H)\cap E(T_s)$ with a vertex $v_e$ that is marked,
and keep the other vertices unmarked. To each such a new vertex $v_e$ 
we attach the weights $\nu(s,e)$ and $\ell(e)$. Using cuts and links,
we can restrict our attention to the red or to the blue subtree,
as needed, and then the operation $\SumTree$ can be used to obtain
\begin{align*}
    \sum_{e\in E(H)\cap E(T_s)}\nu(s,e),~~
    \sum_{e\in E(H)\cap E(T_s)\text{ red}}\ell(e),~~ 
    \sum_{e\in E(H)\cap E(T_s)\text{ blue}}\ell(e).
\end{align*}
The values $\ell(e)$ are never updated. The values $\nu(s,e)$
get updated through operations similar to $\AddTree$ in the red or blue
subtree. However, note that to update $\nu(s, e)$ we need to take
into account the multiplicative weights $\ell(e)$ because, for example,
we have updates of the form $\nu(s,e):=\nu(s,e)+\lambda\cdot \ell(e)$
for all red edges $e$. The multiplicative weight $\ell(e)$ associated
to $v_e$ is constant and attached to $v_e$. Usual dynamic
trees can handle this kind of multiplicatively weighted updates
to a tree.

We now turn our attention to $C_s$. We maintain the dual tree $C_s^*$
using the edge-weighted embedded forest of \cref{lem:embeddedtree},
where, for each $e\in C_s$, we maintain with $e^*$ the two values 
$\nu(s,e)$ and $\ell(e)$.
Here, the marked edges are those dual to $C_s\cap E(H)$. 
The values $\ell(e)$ do not get updated.
To update the values $\nu(s,e)$, we need an update operation in the embedded
forest, where we add the weight $\lambda\cdot \ell(e)$ to all the edges $e$ to
one side of a path. We also need an operation to obtain 
the sum of the weights $\nu(s, e)$ over all the marked edges 
$e\in E(H)\cap C_s$, and to obtain the sum of $\ell(e)$ over
all marked edges $E(H)\cap C_s$ on one side of a path. 
Here it is relevant the property that the red (resp.~blue) edges of $C_s$
are those to the right (resp.~left) of a path in $C^*_s$.
The dynamic forest of \cref{lem:embeddedtree} can be extended 
to have these operations, including the $\ell(e)$-weighted
update. With these operations, we can efficiently obtain 
\begin{align*}
    &\sum_{e\in E(H)\cap E(C_s)}\nu(s,e),~~
    \sum_{e\in E(H)\cap E(C_s)\text{ red}}\ell(e),~~ 
    \sum_{e\in E(H)\cap E(C_s)\text{ blue}}\ell(e).
\end{align*}

With this information, we can obtain the values from \eqref{eq:mean-planar}
because
\begin{align*}
    \sum_{e\in E(H)} \nu(s,e) ~&=~ \sum_{e\in E(H)\cap E(T_s)}\nu(s,e) + 
                                \sum_{e\in E(H)\cap E(C_s)}\nu(s,e),\\
    \sum_{e\in E(H)\text{ red}}\ell(e) ~&=~, 
        \sum_{e\in E(H)\cap E(T_s)\text{ red}}\ell(e) +
        \sum_{e\in E(H)\cap E(C_s)\text{ red}}\ell(e),~\text{~ and}\\ 
    \sum_{e\in E(H) \text{ blue}}\ell(e) ~&=~   
        \sum_{e\in E(H)\cap E(T_s)\text{ red}}\ell(e) + 
        \sum_{e\in E(H)\cap E(C_s)\text{ blue}}\ell(e).
\end{align*}

For each $i\in [k]$, we can obtain the values at $s_i$ from those at
$s_{i-1}$ making $O(1)$ operations in the dynamic forests, and therefore we
need $O(\log n)$ amortized time to sweep $s$ along the $f$-interval 
$I_i= [s_{i-1},s_i]_f$. It follows that we can compute
\[
    \int_{s\in I_i}\int_{q\in \HH} d_\GG(s,q)  \, dq \, ds
\]
in $O(\log n)$ amortized time for each $i\in [k]$.
Adding these values for the intervals $I_i$ contained in $\GG_f\cap \HH$, 
we obtain the following.

\begin{theorem}
    Let $G$ be a planar graph with $n$ vertices, 
    nonnegative edge-lengths, and let $\GG$ be the corresponding continuous graph.
    Let $H$ be a subgraph of $G$ and let $\HH\subseteq \GG$
    be the corresponding continuous subgraph.
    Let $f$ be a face in some given embedding of $G$.
    In  $O(n\log n)$ time, we can compute  the value
    \[
        \int_{s\in \HH\cap \GG_f} \int_{q\in \HH} 
            d_\GG(s,q) \, dq \, ds.
    \]
\end{theorem}

Fix an embedding of a planar graph $G$.
By adding over all faces $f$ of the embedding the value
$\int_{s\in \HH\cap \GG_f} \int_{q\in \HH} d_\GG(s,q) \, dq \, ds$, 
we obtain $2\cdot \sumdist(\HH,\GG)$, because each edge of $H$
is twice on the boundary of a face, in $O(nF\log n)$ time.
From this, it is trivial to compute $\mean(\HH,\GG)$, and we have
shown the part of \cref{thm:planar} concerning the mean distance.


\section{Conclusion}
\label{sec:conclusion}
We presented the first subquadratic algorithms to compute the diameter and the mean distance in continuous graphs, for two non-trivial graph classes.
We expect that the approach for graphs parameterized
by the treewidth can be adapted for computing
other statistics defined by the distance between two points 
selected at random in a continuous subgraph $\HH\subseteq\GG$,
like a cumulative density function (CDF) and higher moments:
\begin{align*}
    &\text{for given $\delta$,  compute} 
        \CDF(\delta,\HH,\GG) =
                        \frac{1}{\ell(\HH)^2}\iint_{p,q\in \HH} 
                            \mathbbm{1}[d_{\GG}(p,q)\le \delta] \,dp \,dq ,\\
    &\text{median distance:}~~\sup\big\{ \delta\in \mathbb{R}_{\ge 0} \mid \CDF(\delta,\HH,\GG)\le 1/2 \big\},\\
    &\text{higher moments, such as }
        \frac{1}{\ell(\HH)^2}\iint_{p,q\in \HH} 
            \big(d_{\GG}(p,q)\big)^2\,dp \,dq.
\end{align*}
The main open question stemming from our work is whether the continuous mean distance
and diameter can be computed in subquadratic time for arbitrary planar graphs.
However, as already mentioned in the introduction, this requires dynamic trees with a set of suitable operations that---at the moment---seem to be out of reach.

\bibliography{bibliography}

\begin{appendix}
\section{A data structure to maintain edge-weighted embedded forest (Proof of \cref{lem:embeddedtree})}
\label{app:proof-data-structure}

In this section we prove \cref{lem:embeddedtree}.
More precisely, we present a data structure to maintain edge-weighted embedded forest with marked edges and 
    operations $\Create$, $\Cut$, $\Link$, $\GetEdgeValue$,
	$\AddLeftPath$ and $\MaxTree$, with  $O(\log n)$ amortized time each.

\begin{proof}[Proof of \cref{lem:embeddedtree}]
	We describe an adaptation of top trees.
	First, we use {\em ternarization}: each vertex of degree at least $4$ is
	replaced by a chain of degree $3$ vertices. See \cref{fig:trees2}.
	The order in the chain is given by the cyclic order of the edges.
	It is important to note that this does not change the set of edges
	that is to the left of a given path. It is cumbersome but standard that
	we have to keep a way to identify edges in the original forest
	and the transformed forest. We keep track at each
	vertex of the cyclic order of its edges.
	The new edges in the forest have weight $-\infty$, and thus
	they never define the maximum and any addition of bounded weights
	does not alter it. From now on, we only discuss how to maintain
	a forest of maximum degree~$3$.
	
	\begin{figure}[h]
	\centering
		\includegraphics[page=2]{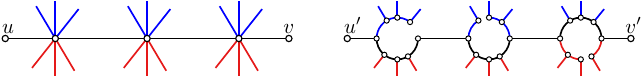}
		\caption{Left: transforming a tree to a tree of maximum degree $3$.}
		\label{fig:trees2}
	\end{figure}

	We build a top tree for each tree in the forest. In the following discussion, 
	we assume familiarity with the presentation of top trees in~\cite{AlstrupHLT05}
	and follow the paradigm used there. Recall that a \emph{path cluster} is 
    a cluster (subtree) with two boundary vertices. 
    For each path cluster $C$, we denote
	by $\pi(C)$ the \emph{cluster path} of $C$, that is, the unique path in $C$
    that connects its two boundary vertices.    
	For each path cluster $C$ and each boundary vertex $v$ of $C$
	let $\tau(C,v)$ be the connected component of $C$
	that contains $v$ after the removal of the edges of $\pi(C)$;
    we call it the \emph{hanging part} of $C$ at $v$.
	Therefore, the edge set of a path cluster $C$ can be broken into
	five parts: the edges on $\pi(C)$, the edges in the hanging part 
	at one boundary vertex, the edges in the hanging part at the other boundary vertex, 
	the edges to the left of $\pi(C)$, and the edges to the right of $\pi(C)$.
	We will keep information for each of them.

    For each cluster $C$ of the top tree, we store the following information, which essentially boils down to two values: the maximum weight among \textit{marked edges} in (part of) 
    the cluster, and an extra value that needs to be added to all edge weights 
    in (part of) the cluster. More precisely:
    
	\begin{itemize}
		\item For each non-path cluster $C$,
			$\maxWeight(C)$ is the maximum weight among the marked edges of~$C$.
		\item For each non-path cluster $C$,
			$\extra(C)$ is an extra additive weight 
			we have to add to all edges of $C$ in proper descendants of $C$.
		\item For each path cluster $C$, $\maxWeightPath(C)$ is 
			the maximum weight among the marked edges of the cluster path $\pi(C)$.
		\item For each path cluster $C$, $\extraPath(C)$ 
            is an extra additive weight we have to add to all edges
			of the path $\pi(C)$ 
			in proper descendants of $C$.
 		\item For each path cluster $C$ and 
            for each boundary vertex $v$ of $C$, $\maxWeightHanging(C,v)$ 
            is the maximum weight among	the marked edges in the hanging part $\tau(C,v)$.
 		\item For each path cluster $C$ and for each boundary vertex $v$ of $C$, 
            $\extraHanging(C,v)$ 
            is an extra additive weight we have to add to all edges
            of $\tau(C,v)$ 
			in proper descendants of $C$. 		
        \item For each path cluster $C$, at its creation we fix an orientation of its cluster path. 
        	This means that we select one boundary point as $\sstart(C)$ and the other as $\eend(C)$. 
        	The cluster path $\pi(C)$ is then oriented always from $\sstart(C)$ to $\eend(C)$.
        	We define $\Left(C)$ as the set of edges in $C$ to the left of $\pi(C)$.
        	Similarly, we define $\Right(C)$.
        	(We clarify that $\Left(C)$ and $\Right(C)$ are {\em not kept} with $C$, 
            it is notation we use to avoid repetitions.)		
		\item For each path cluster $C$, $\maxWeightLeft(C)$ is 
			the maximum weight among all the marked edges in $\Left(C)$.
		\item For each path cluster $C$, $\extraLeft(C)$ 
            is an extra additive weight we have to add to all edges
			of $\Left(C)$ 
			in proper descendants of $C$. 	
		\item For each path cluster $C$, $\maxWeightRight(C)$ is 
			the maximum weight among all the marked edges in $\Right(C)$.
		\item For each path cluster $C$, $\extraRight(C)$ 
            is an extra additive weight we have to add to all edges
			of $\Right(C)$ in proper descendants of $C$.	
	\end{itemize}
	
	\begin{figure}[]
		\centering
		\includegraphics[page=3,width=\textwidth]{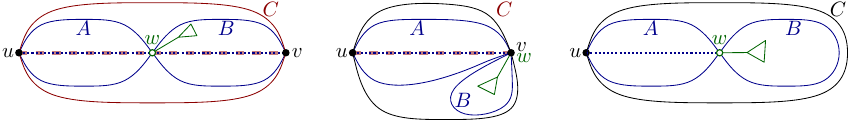}
		\caption{Three of the five cases that appear when a cluster
			is obtained from joining two clusters, that is, when the top node
			$C$ is the parent of $A$ and $B$ in the top tree.
			The red path is the cluster path of $C$ (left and center),
			while the thinner, blue path is a cluster path for the 
			children $A$ and $B$.
			The filled dots are boundary vertices of $C$, the empty dots
			are boundary vertices of the children of $A$ and $B$ that are
			not boundary vertices for $C$. In all cases, $w$ is
			the common boundary vertex of $A$ and $B$.}
		\label{fig:trees3}
	\end{figure}
	
	We now discuss how this information is maintained through the internal operations 
	of top trees. To jump to the interesting cases, we 
	consider the function \textsc{Split}$(C)$, where we split a path cluster $C$
	into its two forming clusters $A$ and $B$; there are two
	cases, shown in the left and center of \cref{fig:trees3}.
	For each path cluster child $D$ ($D=A$ or $D=B$ or both),
	we can detect whether $\pi(C)$ and $\pi(D)$ have the same orientation with the following
	observation:
	$\pi(C)$ and $\pi(D)$ have the same orientation if and only if
	$\sstart(C)= \sstart(D)$ or 
	$\eend(C)=\eend(D)$.
	If $\pi(D)$ and $\pi(C)$ have the same orientation, we update
	\begin{align*}
		\extraLeft(D)&=\extraLeft(D)+ \extraLeft(C),\\
		\extraRight(D)&=\extraRight(D)+ \extraRight(C),\\
		\maxWeightLeft(D)&=\maxWeightLeft(D)+ \extraLeft(C),\\
		\maxWeightRight(D)&=\maxWeightRight(D)+ \extraRight(C),\\
		\extraPath(D)&=\extraPath(D)+\extraPath(C), \text{ and }\\
		\maxWeightPath(D)&=\maxWeightPath(D)+\extraPath(C).
	\end{align*}
	If, on the other hand, $\pi(D)$ and $\pi(C)$ have opposite orientations,
	then we update 
	\begin{align*}
		\extraLeft(D)&=\extraLeft(D)+ \extraRight(C),\\
		\extraRight(D)&=\extraRight(D)+ \extraLeft(C),\\
		\maxWeightLeft(D)&=\maxWeightLeft(D)+ \extraRight(C),\\
		\maxWeightRight(D)&=\maxWeightRight(D)+ \extraLeft(C),\\
		\extraPath(D)&=\extraPath(D)+\extraPath(C), \text{ and }\\
		\maxWeightPath(D)&=\maxWeightPath(D)+\extraPath(C).
	\end{align*}
	
	If in a split operation both $A$ and $B$ are path clusters (\cref{fig:trees3}, left), 
	we have to take care of the edges outside $\pi(C)$ that are incident to $w$ 
	because they are neither to the left nor to the right of $\pi(A)$ or $\pi(B)$.
	This is the green tree incident to $w$ inside $B$ in \cref{fig:trees3}, left.
	Because $w$ has degree at most $3$, there is at most one such an edge $e$ incident
	to $w$. If $e$ is to the left of $\pi(C)$, then for each child $D$ of $C$ we set
	\begin{align*}
		\extraHanging(D,w)&=\extraHanging(D,w)+\extraLeft(C), \text{ and }\\
		\maxWeightHanging(D,w)&=\maxWeightHanging(D,w)+ \extraLeft(C).
	\end{align*}
	Otherwise $e$ is to the right of $\pi(C)$ and we perform
	\begin{align*}
		\extraHanging(D,w)&=\extraHanging(D,w)+\extraRight(C), \text{ and }\\
		\maxWeightHanging(D,w)&=\maxWeightHanging(D,w)+ \extraRight(C).
	\end{align*}
	Finally, for each boundary vertex $u$ of $C$ and each child $D\in \{A,B\}$ 
	with the same boundary vertex, we set 
	\begin{align*}
		\extraHanging(D,u)&=\extraHanging(D,u)+\extraHanging(C,u), \text{ and }\\
		\maxWeightHanging(D,u)&=\maxWeightHanging(D,u)+ \extraHanging(C,u).
	\end{align*}	
			
	Let us discuss now what else has to be done during split
	if $C$ is a path cluster, $A$ is a path cluster, and $B$ is a non-path cluster;
	see the center of \cref{fig:trees3}.
	Let $w=A\cap B$ and let $u$ be the other common boundary of $A$ and $C$.
	We then have to update
	\begin{align*}
		\extra(B)&=\extra(B)+\extraHanging(C,w),\\
		\maxWeight(B)&=\maxWeight(B,w)+\extraHanging(C,w),\\
		\extraHanging(A,u)&=\extraHanging(A,u)+\extraHanging(C,u),\\
		\extraHanging(A,w)&=\extraHanging(A,w)+\extraHanging(C,w),\\
		\maxWeightHanging(A,u)&=\maxWeightHanging(A,u)+\extraHanging(C,u), \text{ and }\\
		\maxWeightHanging(A,w)&=\maxWeightHanging(A,w)+\extraHanging(C,w).
	\end{align*}	

	The other three cases of \textsc{Split}$(C)$ are similar and simpler because
	there is no additional extra to the left or to the right to be kept track of.
	For example, in the case in the right of \cref{fig:trees3}, all the values
	associated to $A$ and $B$ get increased by $\extra(C)$.
	
	The operation \textsc{Join}$(A,B)$ is similar. 
	Let us consider for example the case shown in \cref{fig:trees3}, left,
	where $C$ is a new path cluster obtained by merging two path clusters $A$ and $B$ with $w=A\cap B$,
	$u$ the common boundary vertex of $A$ and $C$, and $v$ the common boundary 
	vertex of $B$ and $C$.
	We set $\extraLeft(C)$,$\extraRight(C)$, $\extraPath(C)$, $\extraHanging(C,u)$ and $\extraHanging(C,v)$ to $0$  and set
	\begin{align*}
		\maxWeightHanging(C,u)&= \maxWeightHanging(A,u),\\
		\maxWeightHanging(C,v)&= \maxWeightHanging(B,v), \text{ and }\\
		\maxWeightPath(C)&= \max\{ \maxWeightPath(A),~ \maxWeightPath(B) \}.
	\end{align*}	
	Then, we look into the left and right sides of $\pi(C)$.
	We can select, for example, $\sstart(C)=u$ and $\eend(C)=v$,
	and initially set $\maxWeightLeft(C)= \maxWeightRight(C)= -\infty$.
	Then, if $\sstart(A)=u$, we have to update
	\begin{align*}
		\maxWeightLeft(C)&= \max\{ \maxWeightLeft(C),~\maxWeightLeft(A)\}, \text{ and }\\
		\maxWeightRight(C)&= \max\{ \maxWeightRight(C),~\maxWeightRight(A)\}.
	\end{align*}	
	Otherwise $\sstart(A)\neq u$ and we update
	\begin{align*}
		\maxWeightLeft(C)&= \max\{ \maxWeightLeft(C),~\maxWeightRight(A)\}, \text{ and }\\
		\maxWeightRight(C)&= \max\{ \maxWeightRight(C),~\maxWeightLeft(A)\}.
	\end{align*}	
	Similarly, if $\eend(B)=v$, we apply
	\begin{align*}
		\maxWeightLeft(C)&= \max\{ \maxWeightLeft(C),~\maxWeightLeft(B)\}, \text{ and }\\
		\maxWeightRight(C)&= \max\{ \maxWeightRight(C),~\maxWeightRight(B)\},
	\end{align*}
	while, if $\eend(B)\neq v$, we set
	\begin{align*}
		\maxWeightLeft(C)&= \max\{ \maxWeightLeft(C),~\maxWeightRight(B)\}, \text{ and }\\
		\maxWeightRight(C)&= \max\{ \maxWeightRight(C),~\maxWeightLeft(B)\}.
	\end{align*}
	Finally, we must take care of the unique edge $e$ attached to $w$ outside of $\pi(C)$, 
    if it exists, and the possible subtree of $e$.
	If $e$ is to the left of $\pi(C)$, then we set $\maxWeightLeft(C)$ to
	\begin{align*}
		\max\{ \maxWeightLeft(C),~\maxWeightHanging(A,w), 
		~\maxWeightHanging(B,w)\}.
	\end{align*}
	(Only one of the last two values is not infinity, because $e$ belongs to only one of the children.)
	If $e$ exists and is not to the left of $\pi(C)$, then we set $\maxWeightRight(C)$ to
	\begin{align*}
		\max\{ \maxWeightRight(C),~\maxWeightHanging(A,w), 
		~\maxWeightHanging(B,w)\}.
	\end{align*}
	An easy way to know whether $e$ is to the left or the right of $\pi(C)$,
    is to store for each cluster the extreme edges of $\pi(C)$. The two
    edges of $\pi(C)$ incident to $w$ can then be obtained from $\pi(A)$ and $\pi(B)$.
    The other cases for join are similar or simpler.
	
	With this, we can implement the desired operations.
	$\Create$ and $\Cut$ are implemented as usual.
	The function $\Link$ is implemented by taking care of the ternarization.
	
	The function $\GetEdgeValue(uv)$ is implemented as  
	$\maxWeightPath(\textsc{Expose}(u,v))$.
	The function $\MaxTree(u)$ is implemented calling
	$\maxWeight(\textsc{Expose}(u))$
	
	The function $\AddLeftPath(\Delta,u,v)$ is implemented
	by calling $\textsc{Expose}(u,v)$, which returns the root node $R$
	of the top tree with $u$ and $v$ as boundary vertices.  
	If $u=\sstart(R)$, then we increase $\extraLeft(R)$
	and $\maxWeightLeft(R)$ by $\Delta$, 
	otherwise we increase $\extraRight(R)$ and $\maxWeightRight(R)$
	by $\Delta$.
\end{proof}

\end{appendix}

\end{document}